\documentclass[11pt, oneside]{article} 
\usepackage[margin=2cm]{geometry}                		
\usepackage[parfill]{parskip}    		
\usepackage{graphicx}			
\usepackage[utf8]{inputenc}
\usepackage[english]{babel} 
\usepackage{amsmath,amssymb,amsthm,bm} 
\usepackage[margin=2cm]{geometry}
\usepackage[color=yellow]{todonotes}
\usepackage{mathrsfs}
\usepackage{url}	
\usepackage{esint}
\usepackage[colorlinks]{hyperref}
\usepackage{enumerate}
\usepackage{outlines}[enumerate]
\usepackage{framed,comment,enumerate}
\usepackage{upgreek}
\usepackage{pgfplots}
\usepackage{float}
\usepackage{tikz,tikz-cd}
\usepackage{rotating}
\usepackage{graphicx}
\usepackage{caption}
\usepackage{multicol}
\usepackage{cancel}
\usepackage{subcaption}
\usepackage[title]{appendix}
\usepackage{tikz-cd}
\usepackage{pgf, pgffor}
\usepackage{subcaption}
\extrafloats{100}
\numberwithin{equation}{section}

\newtheorem{theorem}{Theorem}[section]

\newtheorem{corollary}{Corollary}[section]
\newtheorem{proposition}{Proposition}[section]
\newtheorem{remark}{Remark}[section]

\theoremstyle{definition}

\newcommand{\ob}[1]{\overline{#1}}
\newcommand{\wt}[1]{\widetilde{#1}}
\newcommand{\wh}[1]{\widehat{#1}}
\newcommand{\mb}[1]{\mathbf{#1}}
\newcommand{\mbb}[1]{\mathbb{#1}}

\newcommand{\mc}[1]{\mathcal{#1}}
\newcommand{\bs}[1]{\boldsymbol{#1}}

\newcommand{\mcal}[1]{\mc{#1}}
\newcommand{\scp}[2]{\left<#1\,,\,#2\right>}

\newcommand{\ad}{\operatorname{ad}}

\def\yh{\mathbf{\widehat{y}}}

\def\p{{\partial}}

\def\ep{{\epsilon}}
\def\rmd{{\color{red}{\rm d}}}
\def\ba{{\mathbf{a}}}
\def\bk{{\mathbf{k}}}
\def\bm{{\mathbf{m}}}
\def\bM{{\mathbf{M}}}
\def\bp{{\mathbf{p}}}

\def\bu{{\mathbf{u}}}

\def\bx{{\mathbf{x}}}
\def\bX{{\mathbf{X}}}

\def\bxi{{\boldsymbol{\xi}}}

\def\p{\partial}

\pgfplotsset{compat=1.16}

\begin{document}
	\title{Geometric Mechanics of the Vertical Slice Model \\
 (For Volume 1, Issue 1 of \textit{Geometric Mechanics})}
    \author{Darryl D. Holm, Ruiao Hu\footnote{Corresponding author, email: ruiao.hu15@imperial.ac.uk}, and Oliver D. Street\\
d.holm@imperial.ac.uk, ruiao.hu15@imperial.ac.uk, o.street18@imperial.ac.uk\\
Department of Mathematics, Imperial College London \\ SW7 2AZ, London, UK}
	\date{\today}
	\maketitle
	\begin{abstract}
        
        The goals of the present work are to: (i) investigate the dynamics of oceanic frontogenesis by taking advantage of the geometric mechanics underlying the class of Vertical Slice Models (VSMs) of ocean dynamics; and (ii) illustrate the versatility and utility of deterministic and stochastic variational approaches by deriving several variants of wave-current interaction models which describe the effects of internal waves propagating within a vertical planar slice embedded in a 3D region of constant horizontal gradient of buoyancy in the direction transverse to the vertical pane.
	\end{abstract}

\tableofcontents

\newpage
\section{Introduction}
Oceanic fronts form by extracting the gravitational potential energy available from the `tilting' of buoyancy isoclines, represented as horizontal gradients of buoyancy. In turn, the formation of oceanic fronts facilitates the transfer of kinetic energy to small-scale mixing and transport \cite{GTSM2022,JCM2021}.
In particular, a region of constant horizontal gradient of temperature (implying a corresponding gradient of buoyancy) in a three-dimensional, vertically-stratified fluid governed by the Euler-Boussinesq equations can induce the formation of fronts emerging in vertical planar flows transverse to the direction of constant horizontal gradient of temperature \cite{A-OBdL2019,CH2013,CH2014,YSCMC2017}. 
The present work derives \textit{Vertical Slice Models} (VSMs) to investigate the formation and subsequent evolution of the fronts emerging in a region of \emph{constant} horizontal gradient of temperature. These VSMs may also be useful for bench-marking numerical schemes, since in 2D they can be run quickly on a single workstation.

Driven by their constant transverse horizontal gradient of temperature, the VSMs in \cite{A-OBdL2019,CH2013,CH2014,YSCMC2017} possess a $y$-independent solution structure for the flows within the vertical plane which includes the dynamics of the $y$-component of the fluid velocity transverse to the vertical plane and the Coriolis force. The $y$-independent solution structure for VSM flows is still a solution of the full 3D Euler-Boussinesq fluid equations, as long as the horizontal gradient of the temperature transverse to the vertical slice remains constant. This solution property arises because the pressure gradient within the vertical plane only accesses the $y$-independent part of the temperature as it varies with time and space in the vertical plane. 

The VSM family can be derived in the Euler–Poincaré (EP) framework of symmetry-reduced Lagrangians in Hamilton's variational principle \cite{HMR1998}. The EP framework involves a constrained Hamilton's principle expressed in the Eulerian fluid description. Their derivation in the EP framework establishes the following properties of each member of the VSM family: the Kelvin–Noether circulation theorem, conservation of potential vorticity on fluid parcels, a Lie–Poisson Hamiltonian formulation possessing conserved Casimirs arising from particle-relabelling symmetry, a conserved domain-integrated energy and an associated variational principle satisfied by the equilibrium solutions.
\smallskip

{\bf Aims of the present work.}
New theoretical methods of analysing the dynamics of complex fluids advecting a variety of different co-evolving order parameters have been developing in the Euler–Poincaré (EP) framework during the past twenty years. See, e.g, \cite{H2002,CMR2004,GBR09,GRT2013} for discussions of these new theoretical methods. The EP framework also offers a means of extending the VSMs to include the new methods for complex fluids, as discussed below. 

\color{black}
The new theoretical methods designed for modelling complex fluid flows include the statistical and probabilistic methods needed for data science, as well as the geometric and analytical methods underlying the theory of nonlinear partial differential equations. For example, recent mathematical discoveries have revealed the Poisson structures of the Eulerian description of ideal (non-dissipative) complex fluids such as superfluids, spin glasses, liquid crystals, ferrofluids, etc. derived in \cite{H2002,CMR2004,GBR09,GRT2013}. Complex fluids transport order parameters that co-evolve in the frame of the fluid motion, and the dynamics of these order parameters reacts back to affect the fluid motions that transport them. 

Motivated by these recent mathematical discoveries, we consider applying the variational methods of the geometric mechanics framework \cite{HSS2009} to guide the investigation of the Poisson structures underlying the ocean science of wave and current interactions. We choose this variational approach because of its versatility in deriving the Poisson structures we seek, even in the presence of any random transport that admits the product rule and the chain rule, \cite{H2015,GH2018,CHLN2022}. The inclusion of randomness into nonlinear ocean dynamics has been shown to be useful for uncertainty quantification and data assimilation methods to reduce uncertainty in computational simulations of ocean models, in particular by using the method of Stochastic Advection by Lie Transport (SALT), which preserves the semidirect-product Poisson structures of classical fluids, \cite{H2015, CCHOS18a, CHLN2022}. The Euler-Poincar\'e and Hamilton-Pontryagin versions of Hamilton's principle have been extended from stochastic paths to \textit{geometric rough paths} in \cite{CHLN2022}.

The present work illustrates the versatility and utility of the variational framework of geometric mechanics by deriving several variants of wave-current interaction models that describe internal waves propagating in the vertical slice model (VSM) with transverse flow that was originally derived in \cite{CH2013}. Frontogenesis in this model was demonstrated via numerical simulations in \cite{YSCMC2017} and front formation in its solution behaviour has been analysed recently in \cite{A-OBdL2019}. From the viewpoint of modelling in  ocean dynamics, the problem statement for VSMs is an augmentation of the standard incompressible Euler–Boussinesq equations with a constant gradient of temperature in the transverse direction. 

Here, we will first review the derivation of VSMs based on the Euler-Poincar\'e approach as in \cite{CH2013} and then re-derive it from a variational approach based on a \textit{composition of smooth maps} to describe the order-parameter dynamics taking place in the frame of the fluid motion. Subsequently, we will use the latter method to augment the VSM to enable inclusion of non-Boussinesq effects and also to add the effects of internal gravity waves propagating in the vertical slice.

The variational approach we take in this paper introduces a certain composition of maps (C$\circ$M) into the well-known approach originally due to Clebsch \cite{Clebsch1859}. Remarkably, the C$\circ$M approach produces an `untangled' (block-diagonal) Poisson structure for total momentum and order-parameters, as well as an `entangled' Poisson structure for the fluid momentum alone. The latter Poisson structure exhibits a semidirect-product action of the fluid velocity vector field on the order-parameter phase space, as well as an additional symplectic bracket among the order parameters. 

Mathematically, the `untangled' version of its Poisson structure separates into the sum of a symplectic two-cocycle bracket in the phase-space of order parameters added to the standard semidirect-product Lie-Poisson bracket for classical fluid dynamics, as derived in the Lagrangian framework of Hamilton's variational principle \cite{HMR1998}. By comparing the two versions of the Poisson structures for these fluid equations with advected co-evolving order parameters, one finds that they exhibit a mathematical equivalence that can be written as $T^*Q/G \simeq T^*(Q/G)\oplus \mathfrak{g}^*$, where $\oplus$ denotes the Whitney sum. Perhaps not surprisingly, this duality has persisted and been remarked about throughout the historical investigations of fluid dynamics involving variational principles with a composition of maps. 

The history of Clebsch's approach to the formulation of Hamilton's principle for classical fluid dynamics is reviewed in \cite{S1959}, where it is supplemented by an additional advection constraint due Lin \cite{Lin1963} who introduced it in deriving Landau's two-fluid classical-quantum model of $He_2$ superfluids. This Clebsch approach with the Lin constraint was supplemented further and applied to $He_3$ superfluids with an additional spin order-parameter in \cite{HK1982}, where of course additional two-cocycles were found. Later, it was applied to Yang-Mills charged fluids in \cite{GHK1983} and to magneto-hydrodynamics (MHD) and other fluid plasma models, as well as nonlinear elasticity in \cite{HK1983}. The same augmented approach was applied in deriving fluid equations for special-relativistic plasmas coupled to electromagnetic fields  in \cite{H1987} and also for general-relativistic fluids in \cite{H1985}. Remarkably, the symplectic two-cocycle arising in the Poisson structure for general-relativistic adiabatic fluid dynamics turned out to comprise the Minkowsky metric for space-time $g_{\mu\nu}$ and its canonically conjugate momentum density $\pi_{\mu\nu}$ in the Arnowitt-Deser-Misner theory of general relativity \cite{H1985}.
After this variational approach had been applied to complex fluids in \cite{H2002}, the mathematical essence of the Poisson structures in this approach was proven to arise via Lie group reduction by symmetry with respect to affine actions yielding precisely the $T^*Q/G \simeq T^*(Q/G)\oplus \mathfrak{g}^*$ equivalence as its associated bundle reduction in \cite{CMR2001,GBR09}. This identification and its main mathematical applications for complex fluids are summarised in detail in \cite{GBR09}. A related approach called ``metamorphosis'' was investigated for shape analysis in image registration with applications to computational anatomy in \cite{HTY2009, BGBHR2011}. 

While the mathematical setting of affine Poisson structures for this topic is quite rich, \cite{H2002,CMR2004,GBR09,GRT2013}, the Clebsch-Lin composition of maps variational approach derives the Poisson structure and identifies its corresponding affine Lie group action quite straight-forwardly, as we shall demonstrate later by comparing the original Euler-Poincar\'e derivation of the EB VSM to the Clebsch-Lin approach. 

Hereafter, the Clebsch-Lin approach with affine advection will be referred to simply as composition of maps, as in \cite{HHS2023a, HHS2023c}. The C$\circ$M approach can be applied  for uncertainty quantification of ocean models by using the Stochastic Advection by Lie Transport (SALT) approach applied in \cite{H2015, CCHOS18a, CHLN2022} for the composition of several maps. The C$\circ$M approach is also compatible with traditional approaches in ocean modelling such as Generalised Lagrangian Mean (GLM). In fact, the C$\circ$M approach has been used to provide a stochastic closure for GLM in \cite{HHS2023b} which will be illustrated further in the present work. 

To reprise, the purpose of the present work is to make a concrete application of the stochastic C$\circ$M approach for an enhancement of the VSMs of \cite{CH2013} to include the wave-current interaction (WCI) effects of internal gravity waves (IGW) propagating in a vertical slice of Euler-Boussinesq (EB) flow including a transverse velocity. This work is meant to be presented as explicitly as possible so it can provide a useful foundation for further applications of geometric mechanics in ocean modelling. 

\paragraph{Motivation of the paper:}
In the present work, we consider the interaction of internal waves with Euler-Boussinesq flows in a vertical slice of fluid undergoing three dimensional volume-preserving flow whose motion transverse to a vertical plane advects Lagrangian particle labels that depend linearly on the transverse Eulerian coordinate.
The constant transverse slope, $s$, of the advected Lagrangian particle labels appears as a constant parameter in the slice dynamics along with a new canonically conjugate transverse momentum density, $\pi_T$, dual in the slice to the potential temperature which is advected in three dimensions as a Lagrangian particle label. 
This concept for a VSM with transverse flow was introduced in \cite{CH2013} and its solution behaviour has been simulated computationally in \cite{YSCMC2017} and analysed recently in \cite{A-OBdL2019}.

\paragraph{Main goals of the paper:}
The ocean modelling goal of the present paper is to include IGW motion in the VSMs with transverse flow introduced in \cite{CH2013}. The mathematical goal accompanying the goal for data assimilation is also to determine the Poisson/Hamiltonian structure of the resulting model, and thereby formulate a new model of stochastic parameterisation of advective transport of potential use for quantifying uncertainty in ocean model simulations. 

\paragraph{Plan of the paper:}
\begin{itemize}
    \item
Section \ref{sec: CH2013 VSM} recalls the EB VSM given in \cite{CH2013} and considers two Euler-Poincar\'e derivations of the VSMs. 
\begin{itemize}
    \item Section \ref{sec: Lag VSM} provides two equivalent derivations of the EB VSM model treated in \cite{CH2013}. Namely, they are the derivation based on affine Euler-Poincar\'e reduction and the Clebsch-Lin derivation in the C$\circ$M context.
    \item Section \ref{sec: Ham VSM} treats the `Entangled' and `Untangled' versions of the Hamiltonian formulation for the vertical slice model.
    \item
    Section \ref{sec: CompSims} demonstrates some of the solution behaviours exhibited by the VSM derived in the previous sections by showing the result of computational simulations. 
\end{itemize}

\item Section \ref{sec:WMFI} treats a further asymptotic expansion estimates of IGW interaction dynamics folowing the wave mean flow interaction (WMFI) closure due to \cite{GH1996} in which the Brunt-V\"ais\"al\"a buoyancy frequency is determined via the Hessian of the fluid pressure. 

\begin{itemize}
    \item
Section \ref{sec:EB_VSM_IGW} treats the Clebsch-Lin derivation of the VSM in the Euler-Boussinesq approximation with the inclusion of a standard WKB model for the IGW Hamiltonian and assumes the existence of a constant Brunt-V\"ais\"ala (BV) buoyancy frequency. 
    \item 
Section \ref{sec: Geom GLM} reviews the geometry of the generalised Lagrangian mean approach to wave mean flow interaction.

    \item 
Section \ref{sec:GLM_vertical_slice} introduces a wave mean flow decomposition of dynamics in the vertical slice of an Euler-Boussinesq fluid.
\end{itemize}

    \item
Section \ref{sec:SALTyVSM} then formulates the stochastic parameterisation model for VSMs based on Stochastic Advection by Lie Transport (SALT). 
    \item 
Section \ref{sec: Conclude Outlook} Concludes with a discussion of the goals achieved in the paper and a survey of potential future work.  
    \item 
Appendix \ref{appendix:expansion} contains an asymptotic expansion which reveals the form of the action for the dynamical system studied in Section \ref{sec:EB_VSM_IGW}.
    \item
Appendix \ref{app-B-Dispersion} contains the derivation of a new dispersion relation for inertial gravity waves in a vertical slice domain.

\end{itemize}

At each level of approximation in the sequence of sections in this paper, we identify the opportunities where stochasticity, e.g., SALT, may be introduced that will preserve the geometric Poisson structure of the model. The dynamical effects of the introduction of SALT into the deterministic VSMs derived here and their implications for data science and data assimilation are beyond the scope of the present paper. However, we expect that the framework for applications of the present approach for data calibration, uncertainty quantification and data assimilation considered previously in \cite{CCHOS18a, CCHOS20} will certainly carry over for the stochastic IGW effects considered here.

\section{Vertical slice models (VSMs)}\label{sec: CH2013 VSM}

\begin{figure}[H]
\centering
\begin{tikzpicture}
\pgfmathsetmacro{\cubex}{4}
\pgfmathsetmacro{\slicex}{2}
\pgfmathsetmacro{\cubey}{2}
\pgfmathsetmacro{\cubez}{3}
\draw[blue,fill=blue!20!white] (0,0,0) -- ++(-\cubex,0,0) -- ++(0,-\cubey,0) -- ++(\cubex,0,0) -- cycle;
\draw[blue,fill=blue!20!white] (0,0,0) -- ++(0,0,-\cubez) -- ++(0,-\cubey,0) -- ++(0,0,\cubez) -- cycle;
\draw[blue,fill=blue!20!white] (0,0,0) -- ++(-\cubex,0,0) -- ++(0,0,-\cubez) -- ++(\cubex,0,0) -- cycle;
\draw[black,fill=blue] (-\slicex,0,0) -- ++(0,0,-\cubez) -- ++(0,-\cubey,0) -- ++(0,0,\cubez) -- cycle;
\draw[blue,dashed] (0,-\cubey,-\cubez) -- (-\cubex,-\cubey,-\cubez);
\draw[blue,dashed] (-\cubex,-\cubey,0) -- (-\cubex,-\cubey,-\cubez);
\draw[blue,dashed] (-\cubex,0,-\cubez) -- (-\cubex,-\cubey,-\cubez);
\draw[->] (0.5*\cubey,-\cubey,-\cubez)--(0.5*\cubey+0.2*\cubey,-\cubey,-\cubez) node[right]{\footnotesize$y$};
\draw[->] (0.5*\cubey,-\cubey,-\cubez)--(0.5*\cubey,-\cubey + 0.2*\cubey,-\cubez) node[above]{\footnotesize$z$};
\draw[->] (0.5*\cubey,-\cubey,-\cubez)--(0.5*\cubey,-\cubey,-\cubez + 0.3*\cubey) node[below left]{\footnotesize$x$};
\draw[-stealth] (0,0.5*\cubey,-\cubez) node[right] {$M \subset \mathbb{R}^2$} .. controls (-0.25*\cubex+ 0.1*\cubex,0.6*\cubey,-0.75*\cubez-0.1*\cubex) and (-0.25*\cubex-0.1*\cubex,0.6*\cubey,-0.75*\cubez+0.1*\cubex)  ..(-0.5*\cubex,0,-0.5*\cubez);
\end{tikzpicture}
\caption{The vertical slice domain}
\label{fig:slice_domain}
\end{figure}
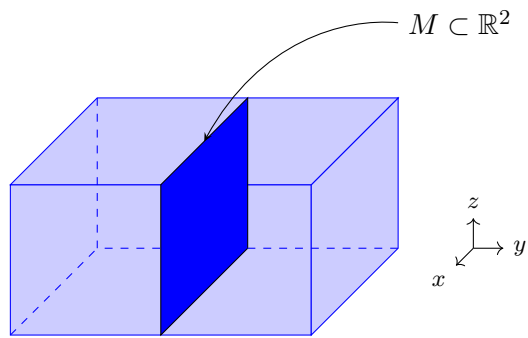
The vertical slice model, as proposed by Eady in the 1940s \cite{E1948,E1949,V2017}, was born from a desire to understand instabilities in atmospheric sciences. It has since proven to be a valuable model for testing numerical methods for geophysical flows, as well as remaining useful for understanding thermal front formation and propagation. Dynamics in a vertical slice model occurs within a vertical ($x$-$z$) plane, denoted by $M$ in Figure \ref{fig:slice_domain}. The slice model features a velocity field tangent to the slice that advects fluid material variables within the slice, as well as a velocity perpendicular to the slice representing transverse flow crossing through the slice at each point. In this section, the vertical slice model will be derived from a variational principle in the Eulerian representation and investigated further by passing to its Hamiltonian formulation. The key idea is to regard the tangential flow in the slice and the perpendicular flow transverse to the slice as the composition of two flow maps. 

\subsection{Lagrangian formulation of VSMs}\label{sec: Lag VSM}

Vertical slice models can be derived using variational principles that apply the Lie group structure of their transverse flow explicitly \cite{CH2013}. In general, the flow map of 3D fluids can be modelled by the diffeomorphism group $\operatorname{Diff}(\mathbb{R}^3)$. Let $g_t \in \operatorname{Diff}(\mathbb{R}^3)$, its action on the coordinates can be written as
\begin{equation*}
    g_t(X,Y,Z) = (x_t(X,Y,Z),\ y_t(X,Y,Z),\ z_t(X,Y,Z))\,,
\end{equation*}
where $(X,Y,Z)$ are the initial Lagrangian particle labels at time, $t_0$, and $(x_t,y_t,z_t)$ denote the Eulerian location of the particle labels $(X,Y,Z)$ at a later time, $t$. In the case of VSMs in 3D, their flow map takes the particular form
\begin{equation}\label{eqn:vertical_slice_flow_map}
    \phi_t(X,Y,Z) = (x_t(X,Z),\ y_t(X,Z) + Y,\ z_t(X,Z))\,,\quad\hbox{which implies}\quad \frac{\p\phi_t}{\p Y} = \begin{pmatrix} 0 \\ 1 \\ 0 \end{pmatrix} \,,
\end{equation}
such that the flow in vertical ($x$-$z$) plane is independent of $Y$ and the flow in the $y$ plane is linear in $Y$. 
One may directly verify that diffeomorphisms of the form \eqref{eqn:vertical_slice_flow_map} constitute to a closed subgroup of $\operatorname{Diff}(\mathbb{R}^3)$ and it is isomorphic to the semidirect-product group $S := {\rm Diff}(M)\ltimes \mcal{F}(M)$. Here, we have $M \subseteq \mbb{R}^2$ as the domain in the vertically sliced plane with coordinates $(x,z)$, $\mcal{F}(M)$ as the space of smooth functions on $M$, and the symbol $\ltimes$ as the semidirect-product \cite{CH2013,HMR1998}. The Lie algebra corresponding to this group is $\mathfrak{s} := \mathfrak{X}(M)\ltimes \mcal{F}(M)$, where $\mathfrak{X}(M)$ is the space of vector fields on $M$. Our Lie algebra consists of elements of the form $(\bu_S,u_T)$, which can be interpreted as components of the velocity \emph{within} and \emph{transverse to} the slice, respectively.

In VSMs, there are two types of advected quantities, belonging to spaces denoted by $V^*_1$ and $V^*_2$, which are acted on by the configuration group $S$ and are assumed to be subspaces of tensor fields $\mathfrak{T}(M)$. Let $(\phi, v)\in S$ and $a_1 \in V^*_1$ be arbitrary, we define the contragradient representation of $S$ on $V^*_1$ to be the pullback by the component of $S$ belonging to $\operatorname{Diff}(M)$
\begin{equation}
    \rho^*_{(\phi, v)^{-1}}a_1 := \rho^*_{\phi^{-1}}a_1 := \phi^*a_1 \,.
\end{equation}
Let $(u_S, u_T) \in \mathfrak{s}$, the infinitesimal action induced by the above representation is given by
\begin{equation}
    \rho^*{}'_{u_S}a_1 := \frac{d}{dt}\bigg|_{t=0} \rho^*_{\exp((-t u_S,-t u_T))}a_1 = \frac{d}{dt}\bigg|_{t=0} \rho^*_{\exp(-t u_S)}a_1 = \mathcal{L}_{u_S} a \,, 
\end{equation}
where $\mathcal{L}_{\fbox{}}$ is the Lie derivative. Since the induced infinitesimal action is independent of $u_T$, we interpret quantities in $V^*_1$ as purely advected by the diffeomorphism group $\operatorname{Diff}(M)$ in the slice domain $M$. The dual operation to the Lie derivative is the diamond operation, $\diamond$, defined through two different pairings as
\begin{equation}\label{eqn:diamond}
    -\scp{b_1 \diamond a_1}{u_S}_{\mathfrak{X}(M)^* \times \mathfrak{X}(M)} = \scp{\mathcal{L}_{u_S} a_1}{b_1}_{V^*_1 \times V_1}\,,
\end{equation}
for all $b_1 \in V_1$. We note that the form of the diamond operation, like the Lie derivative, depends on the concrete choices of the vector space $V_1$ according to its infinitesimal transformation under the Lie derivative.
In VSMs, the fluid quantity that is acted on by $S$ via the contragradient representation is the volume form, $D\,d^2x \in \operatorname{Den}(M) \subset V^*_1$, which is assumed to be advected. That is, the dynamics of $D\,d^2x$ is characterised by the pushforward relation
\begin{equation}
    D_t\,d^2x_t = \rho^*_{(\phi_t, v_t)}(D_0\,d^2x_0) = \phi_* (D_0\,d^2x_0)\,,
\end{equation}
where $D_0 \,d^2x_0$ is the initial reference volume form. Defining the vector field $\bu_S \cdot\nabla = \dot{\phi}\phi^{-1} \in \mathfrak{X}(M)$ where $\nabla := (\p_x ,\p_y)$ is the gradient operator in the slice domain, we have that the pushforward relation is the solution to the advection equation
\begin{equation}\label{eq:D advection}
    \p_t (D\,d^2x) + \mathcal{L}_{u_S} (D\,d^2x) = 0 \quad \Rightarrow\quad \p_t D + \nabla\cdot \left(\bu_S D\right) = 0\,.
\end{equation}
In relation to the advection of the 3D volume form $D\,d^3x$ by the full three dimensional flow, one can obtain the advection of in slice volume form $D\,d^2x$ in the domain $M$ under the assumptions that both $D$ and the transverse component of the 3D velocity field $u_T$ are independent of the Lagrangian label $Y$.

Similarly to $V^*_1$, we can define the contragradient representation of $S$ on $V^*_2$ via the pullback by the $\operatorname{Diff}(M)$ component of $S$. Let $a_2 \in V^*_2$, then the contragradient representation is $\rho^*_{(\phi, v)^{-1}}a_2 := \phi^* a_2$. Additionally, we define the affine representation of $S$ on $V^*_2$ as 
\begin{equation}
    \Phi_{(\phi, v)}a_2 := \rho^*_{(\phi, v)^{-1}}a_2 + c((\phi, v)) = \rho^*_{\phi^{-1}}a_2 + c((\phi, v)) = \phi^*a_2 + c((\phi, v))\,,
\end{equation}
where $c \in \mcal{F}(S, V^*_2)$ is the group one-cocyle satisfying 
\begin{equation}
    c((\phi_1, v_1)\cdot (\phi_2, v_2)) = \rho^*_{(\phi_2, v_2)^{-1}} c((\phi_1, v_1)) + c((\phi_2, v_2))\,,
\end{equation}
and $\cdot$ denotes the group action of the semidirect-product group $S$ on itself. The infinitesimal action induced by the action by affine representation $\Phi$ is given by
\begin{align}\label{eq:inf affine action}
\begin{split}
    \Phi'_{(u_S, u_T)}a_2 &:= \frac{d}{dt}\bigg|_{t=0} \Phi_{\exp((t u_S, t u_T))}a_2 \\
    &= \frac{d}{dt}\bigg|_{t=0} \rho^*_{\exp(-t u_S)}a_2 + c(\exp(t u_S, t u_T)) = \mathcal{L}_{u_S} a_2 + {\rm d}c((u_S, u_T))\,,
\end{split}
\end{align}
where ${\rm d}c \in \mcal{F}(\mathfrak{s}, V^*_2)$ is the Lie-algebra two-cocycle associated with the group one-cocycle and it is defined by ${\rm d}c = T_{(e, 0)} c$. The dual operator to ${\rm d}c$ is ${\rm d}c^T \in \mcal{F}(V, \mathfrak{s}^*)$ and it is defined by
\begin{equation}
    \scp{{\rm d}c^T(b_2)}{(u_S, u_T)}_{\mathfrak{s}^*\times \mathfrak{s}} := \scp{{\rm d}c((u_S, u_T))}{b_2}_{V^*_2\times V_2}\,,
\end{equation}
for all $b_2 \in V_2$. Then, the dual operator to the induced infinitesimal action described in equation \eqref{eq:inf affine action} to defined be
\begin{align}
    \scp{\mcal{L}_{u_S}a_2 + {\rm d}c((u_S, u_T))}{b_2}_{V_2^*\times V_2} = \scp{-b_2 \diamond a_2}{u_S}_{\mathfrak{X}(M)^*\times \mathfrak{X}(M)} + \scp{{\rm d}c^T(b_2)}{(u_S, u_T)}_{\mathfrak{s}^* \times \mathfrak{s}} \,.
\end{align}
Note that the diamond operator in the above definition is operationally identical to that found in equation \eqref{eqn:diamond}, so we have therefore reused the notation. For the remainder of the paper, we will employ the diamond notation to universally refer to the dual operator to the Lie derivative. In VSMs, fluid quantity that is assumed to be acted on by $S$ via the affine action is the potential temperature scalar $\vartheta_s \in \mathcal{F}(M) \subset V^*_2$. That is, the dynamics of $\vartheta_s$ is characterised by the affine group action
\begin{equation}\label{eq:affine vartheta action}
    \vartheta_s(t) = \Phi_{(\phi, v)}\vartheta_s(0) = \phi_* \vartheta_{s}(0) + c((\phi, v)^{-1})\,,
\end{equation}
where $\vartheta_s(0)$ is the initial reference volume form. The one-cocycle relevant to the VSMs is given by 
\begin{equation}\label{eq:VSM one-cocycle}
    c((\phi, v)) := c_2(v) := sv \,,
\end{equation}
where $s \in \mbb{R}$ is a constant. One may directly verify the one-cocycle identity is satisfied 
\begin{align*}
    &\rho^*_{(\phi_2, v_2)^{-1}} c((\phi_1, v_1)) + c((\phi_2, v_2)) = \rho^*_{(\phi_2, v_2)^{-1}}sv_1 + s v_2 = s\phi_2^*v_1 + s v_2\,,\\
    &c((\phi_1, v_1)\cdot (\phi_2, v_2)) = c((\phi_1 \phi_2, v_1\phi_2 + v_2)) = s\phi_2^*v_1 + s v_2\,.
\end{align*}
The Lie-algebra two-cocycle ${\rm d}c$ and its dual ${\rm d}c^T$ can be calculated to be
\begin{equation}
    \frac{d}{dt}\bigg|_{t=0} c(\exp(tu_S, tu_T)) = \frac{d}{dt}\bigg|_{t=0} s \sum_{n=1}^{\infty} \sum_{m=0\,,\,i<n}^{\infty}\frac{1}{n!}t^n\mathcal{L}_{u_S}^m u_T = s u_T\,, \quad {\rm d}c^T(b_2) = (0, s b_2)\,.
\end{equation}
Defining the vector field $\bu_S \cdot \nabla = \dot{\phi}\phi^{-1} \in \mathfrak{X}(M)$ and scalar function $u_T = \phi_* \dot{v} \in \mathcal{F}(M)$, we have affine advection equation for $\vartheta_s$
\begin{equation}\label{eq:affine advection}
    \p_t \vartheta_s + \mathcal{L}_{u_S} \vartheta_s + {\rm d}c((u_S, u_T)) = 0 \quad \Rightarrow\quad \p_t \vartheta_s + \bu_S \cdot \nabla \vartheta_s + s u_T = 0\,,
\end{equation}
where the solution is given by $\vartheta_s(t) = \phi_{t*}\vartheta_s(0) - s\phi_{t*}v_t$.
In relation to the temperature scalar $\vartheta$ in 3D which is advected by the full three dimensional flow, one can obtain the affine advection equation \eqref{eq:affine vartheta action} by assuming that $\vartheta$ has a constant derivative in the $y$ direction. That is, $\vartheta$ has the following decomposition
\begin{equation}\label{eqn:vertical_slice_advected_scalars}
    \vartheta(x,y,z,t) = \vartheta_s(x,z,t) + (y - y_0)s \,,
\end{equation}
where $s \in \mbb{R}$ is the same constant appearing in group one-cocycle definition \eqref{eq:VSM one-cocycle}.

To obtain the full VSMs Euler-Poincar\'e equation, we use the following affine Euler-Poincar\'e variational principle.
\begin{theorem}[Affine Euler-Poincar\'e theorem for VSMs]\label{thm:EP_slice}
    Recalling the preceding dynamical variables for the VSMs, namely the slice vector field $u_S \in \mathfrak{X}(M)$, the transverse velocity scalar $u_T \in \mcal{F}(M)$, the in slice volume form $D\,d^2x \in \operatorname{Den}(M)$ and the potential temperature scalar $\vartheta_s \in \mcal{F}(M)$. 
    Suppose we have a Lagrangian $\ell(u_S, u_T,D,\vartheta_s): (\mathfrak{X}(M)\ltimes \mcal{F}(M))\times\operatorname{Den}(M)\times \mcal{F}(M) \mapsto \mathbb{R}$. An application of Hamilton's Principle
    \begin{equation*}
        0 = \delta S = \delta \int_{t_0}^{t_1} \ell(u_S, u_T,{D},\vartheta_s)\,dt \,,
    \end{equation*}
    subject to the following constraints derived from their definitions 
    \begin{align*}
        (u_S, u_T) := (\dot{\phi}\phi^{-1}, -\phi_*\dot{v}) \quad &\Rightarrow \quad \delta(u_S,u_T) = (\dot{v}_S, \dot{v}_T) - \ad_{(u_S,u_T)}(v_S,v_T)
        \,,\\
        D\,d^2x := \phi_*(D_0\,d^2x_0)\quad &\Rightarrow \quad \delta {D}\,d^2x = -\mcal{L}_{v_S}({D}\,d^2x)
        \,,\\
        \vartheta_s(t) := \phi_*\vartheta_s(0) - s\phi_*v  \quad &\Rightarrow \quad \delta\vartheta_s = -\mcal{L}_{v_S}\vartheta_s - sv_T
        \,,
    \end{align*}
    where $(v_S,v_T) := (\delta \phi\phi^{-1}, -\phi_* \delta v) \in \mathfrak{s}$ are arbitrary Lie-algebra elements, implies the Euler-Poincar\'e equations
    \begin{align}\label{eqn:abstract_vertical_slice_EP}
    \begin{split}
        \left(\p_t + \ad^*_{u_S} \right)\frac{\delta\ell}{\delta u_S} &= \frac{\delta\ell}{\delta u_T} \diamond u_T + \frac{\delta\ell}{\delta {D}} \diamond {D} + \frac{\delta\ell}{\delta \vartheta_s} \diamond \vartheta_s \,,\\
        \left(\p_t + \mcal{L}_{u_S} \right)\frac{\delta\ell}{\delta u_T} &= - \frac{\delta \ell}{\delta \vartheta_s}s\,. 
    \end{split}
    \end{align}
\end{theorem}
\begin{remark}
    For the avoidance of doubt, the variational derivative of the Lagrangian with respect to the volume form $D\,d^2x$ is denoted by $\delta\ell / \delta D$ and is a $0$-form on $M$. As is clearer from the notation, $\delta\ell / \delta \vartheta_s$ and $\delta\ell / \delta u_T$ are $2$-forms (densities) on $M$ and $\delta\ell / \delta u_S$ is a $1$-form density.
\end{remark}
\begin{proof}
    We verify this by direct calculation. 
    \begin{align}
        \begin{split}
            0 &= \delta \int_{t_0}^{t_1} \ell(u_S, u_T,{D},\vartheta_s)\,dt = \int_{t_0}^{t_1} \scp{\frac{\delta \ell}{\delta (u_S,u_T)}}{\delta (u_S,u_T)}  + \scp{\frac{\delta \ell}{\delta D}}{\delta D} + \scp{\frac{\delta \ell}{\delta \vartheta_s}}{\delta \vartheta_s}\,dt \\
            &= \int_{t_0}^{t_1} \scp{\frac{\delta \ell}{\delta (u_S, u_T)}}{(\dot{v}_S, \dot{v}_T) - \ad_{(u_S,u_T)}(v_S,v_T)} + \scp{\frac{\delta \ell}{\delta D}}{-\mcal{L}_{v_S}({D}\,d^2x)} \\
            & \qquad \qquad \qquad + \scp{\frac{\delta \ell}{\delta \vartheta_s}}{-\mcal{L}_{v_S}\vartheta_s - sv_T}\,dt\\
            &= \int_{t_0}^{t_1} \scp{-\left(\p_t + \ad^*_{u_S}\right)\frac{\delta \ell}{\delta u_S} + \frac{\delta \ell}{\delta u_T}\diamond u_T}{v_S} + \scp{-\left(\p_t + \mathcal{L}_{u_S}\right)\frac{\delta \ell}{\delta u_T}}{v_T} \\
            & \qquad \qquad \qquad + \scp{\frac{\delta \ell}{\delta D}\diamond ({D}\,d^2x)}{v_S} + \scp{\frac{\delta \ell}{\delta \vartheta_s}\diamond \vartheta_s}{v_S} + \scp{-s\frac{\delta \ell}{\delta \vartheta_s}}{v_T}\,dt \\
            &= \int_{t_0}^{t_1} \scp{-\left(\p_t + \ad^*_{u_S}\right)\frac{\delta \ell}{\delta u_S} + \frac{\delta \ell}{\delta u_T}\diamond u_T + \frac{\delta \ell}{\delta D}\diamond ({D}\,d^2x) + \frac{\delta \ell}{\delta \vartheta_s}\diamond \vartheta_s}{v_S}  \\
            & \qquad \qquad \qquad + \scp{-\left(\p_t + \mathcal{L}_{u_S}\right)\frac{\delta \ell}{\delta u_T}-s\frac{\delta \ell}{\delta \vartheta_s}}{v_T}\,dt\,.
        \end{split}
    \end{align}
    Then, one have the desired affine Euler-Poincar\'e equations \eqref{eqn:abstract_vertical_slice_EP} from the fact that $v_S$ and $u_T$ are assumed to be arbitrary. 
\end{proof}

\begin{remark}\label{rmk:general_vertical_slice_EP}
    Together with the evolution equations of $D\,d^2x$ and $\vartheta_s$, equations \eqref{eq:D advection} and \eqref{eq:affine advection} respectively, the Euler-Poincar\'e equations corresponding to Theorem \ref{thm:EP_slice} may be written as
    \begin{align}
        (\p_t + \mathcal{L}_{u_S})\left( \frac{1}{{D}}\frac{\delta\ell}{\delta u_S} \right) + \frac{1}{{D}}\frac{\delta\ell}{\delta u_T}du_T &= d\left( \frac{\delta\ell}{\delta {D}} \right) -\frac{1}{{D}}\frac{\delta\ell}{\delta\vartheta_s}d\vartheta_s
        \label{eqn:EP_slice_mom}
        \,,\\
        (\p_t + \mathcal{L}_{u_S})\left( \frac{1}{{D}}\frac{\delta\ell}{\delta u_T} \right) &= - \frac{1}{{D}}\frac{\delta\ell}{\delta\vartheta_s}s
        \label{eqn:EP_slice_transverse}
        \,,\\
        (\p_t + \mathcal{L}_{u_S})\vartheta_s + u_Ts &= 0
        \label{eqn:EP_slice_advection_scalar}
        \,,\\
        (\p_t + \mathcal{L}_{u_S})({D}\,d^2x) &= 0 
        \label{eqn:EP_slice_advection_volume_form}
        \,.
    \end{align}
\end{remark}

\begin{corollary}[The Euler-Boussinesq Eady model]
    An application of Theorem \ref{thm:EP_slice} to the Lagrangian
    \begin{equation}\label{eqn:Eady_action}
        \ell[u_S,u_T,{D},\vartheta_s,p] = \int_M \frac{{D}}{2}(|\bu_S|^2 + u_T^2) + {D}fu_Tx + \frac{g}{\vartheta_0}{D}\left( z - \frac{H}{2} \right) \vartheta_s + p(1-{D})\, d^2x \,,
    \end{equation}
    yields the Euler-Boussinesq Eady equations \cite{CH2013},
    \begin{align}
	   \p_t \bu_S + \bu_S \cdot \nabla \bu_S &= fu_T\wh{x}+ \frac{g}{\vartheta_0}\vartheta_s\wh{z} - \nabla p \,,
	   \label{eqn:EBE-slice}\\
	   \p_t u_T + \bu_S \cdot \nabla u_T &= - f\bu_S\cdot\wh{x}
        -\frac{g}{\vartheta_0}\left( z - \frac{H}{2} \right)s
	   \label{eqn:EBE-transverse}\,,\\
	   \p_t\vartheta_s + \bu_S \cdot \nabla \vartheta_s + u_Ts &= 0
	   \label{eqn:EBE-temperature}\,,\\
	   \nabla \cdot \bu_S &= 0
	   \label{eqn:EBE-incompressibility}\,.
    \end{align}
\end{corollary}

\begin{remark}[In what sense is the VSM for the Euler-Boussinesq Eady equations three dimensional?]
All of the equations in the EB Eady model above are evaluated in two dimensions on the vertical slice. 
In what sense, then, does the transverse velocity $u_T$ confer any sense of three dimensionality? The answer 
to this question stems from the constant slope $s$ of the advected scalar function $\vartheta$ in  \eqref{eqn:vertical_slice_advected_scalars}. Consider 3D advected Lagrangian labels $(L_1(x_1,x_2,t), L_2(x_1,x_2,t),L_3(x_1,x_2,x_3,t))$ with $(x_1,x_2,x_3)=(x,z,y)$ as in \eqref{eqn:vertical_slice_advected_scalars}. In that case, the 3rd row of the Jacobian matrix $J_{ij}=\p L_i/\p x_j$ for the inverse map (Euler-to-Lagrange) would have entries $J_{3j}= (\p L_1/\p x_3,\p L_2/\p x_3, \p L_3/\p x_3)=(0,0,s)$. Consequently, imposing the relation $\det J = 1$ in 3D would imply the relation $\det J = 1/s$ on the 2D vertical slice. Hence, the 2D velocity would be divergence-free, i.e., $\nabla \cdot \bu_S = 0$ in equation \eqref{eqn:EBE-incompressibility} and the transverse velocity $u_T$ would confer a sense of three dimensionality for the class of flows with this type of Lagrangian label dependence for the inverse flow map. In particular, the corresponding Lagrange-to-Euler 3D flow map is written in equation \eqref{eqn:vertical_slice_flow_map}.
\end{remark}

\subsection{Clebsch-Lin derivation of VSMs}

In the above summary of the Lagrangian formulation of VSMs, the actions required to produce the Euler-Poincar\'e equations are nonstandard. That is, the structure is more involved than that required for standard models of continuum dynamics, the equations for which can be derived from the Clebsch approach, Hamilton-Pontryagin approach, or by using Lin-constrained variations. Thus, we seek to describe VSMs without the derivation relying on this action, and instead constrain the required relationships using a Lagrange multiplier. The formulation undertaken here rederives the VSM by using a mixed Clebsch-Lin variational principle, thus deriving the Euler-Boussinesq Eady model without the need for the more involved geometry presented in the previous section. This can be thought of as a `short cut' to the above derivation, and is more usable in practice but neglects to reveal the intricacies of the underlying geometry. Note that the resulting equations have the same Lie-Poisson structure as those derived from the composition-of-maps (C$\circ$M) approach for Hamilton's variational principle for the symmetry reduced Eulerian representation of fluid dynamics \cite{HHS2023c}.

\begin{remark}\label{rmk:Clebsch-Lin}
    We have introduced the term `Clebsch-Lin' for this approach, since the variational principle will use a combination of the constrained Euler-Poincar\'e variations and the Clebsch approach of constraining advected equations on passively advected quantities. Namely, the variations of the two dimensional vector field, $u_S$, and the advected quantity, $Dd^2x$, which is advected purely by $u_S$, will be constrained to take the form of the Lin constraints found in the standard theory of Euler-Poincar\'e reduction for semidirect product spaces \cite{HMR1998}. Recall that the variable representing the physics of the temperature, $\vartheta_s$, is not advected purely by the vector field $u_S$. The equation satisfied by $\theta_s$ will be enforced by a Lagrange multiplier as a Clebsch constraint, and the variation of $\theta_s$, as well as the variations of the remaining variables, will be taken to be arbitrary.
\end{remark}

We begin by formulating the Clebsch-Lin form of Hamilton's principle corresponding to the EB Eady VSM \eqref{eqn:EP_slice_mom} - \eqref{eqn:EP_slice_advection_volume_form} by augmenting the action to the sum of the EB Eady VSM Lagrangian \eqref{eqn:Eady_action} and a constraint on the dynamics of the transverse velocity $u_T$. The variational principle then reads 
\begin{align}
\begin{split}
    0 = \delta S[u_S, u_T, D, \vartheta_s, p, \pi_T] = &\delta \int_a^b \int_{M}
\tfrac12 {D} |\bu_S|^2 + \tfrac12 {D}u_T^2 + {D} u_Tfx + \frac{g}{\vartheta_0} D\left(z - \frac{H}{2} \right)\vartheta_s - p({D}-1)
\\&\hspace{2cm}
- \pi_T \left( \p_t \vartheta_s + \bu_S\cdot\nabla \vartheta_s  + su_T\right) 
d^2x\, dt \,. \label{eqn: Cleb Lin VSM}
\end{split}
\end{align}
Here, following Remark \ref{rmk:Clebsch-Lin}, we have the constrained Euler-Poincar\'e variations $\delta u_S = \p_t\xi - {\rm ad}_{u_S}\xi$ and $\delta ({D} d^2x_s)=-\,{\cal L}_{\xi}({D} d^2x_s)$ where $\xi\in\mathfrak{X}(M)$ are arbitrary and vanishes at the boundaries. Furthermore, all other variations, $\delta \pi_T, \delta u_T, \delta \vartheta_s$ and $\delta p$ are taken to be arbitrary. In \eqref{eqn: Cleb Lin VSM}, the Lagrange multiplier $\pi_T$ is introduced which takes the form of momentum density and it is identified as the total momentum normal to an area element on the slice in the horizontal direction transverse to the vertical slice.
Variations of the action produces the following variational derivatives 
\begin{align}
    \begin{split}
        0 = &\int_a^b \int_{M}
B \,\delta {D} 
+ \delta \bu_S \cdot \left( {D} \bu_S - \pi_T \nabla \vartheta_s \right) 
+ \delta u_T  \left( {D} (u_T + fx) - s\pi_T \right) 
- \delta p ({D}-1)
\\&\hspace{15mm}
+ \delta \vartheta_s \left( \p_t \pi_T + {\rm div}(\pi_t \bu_S) + {D}\gamma(z) \right)
- \delta \pi_T \left( \p_t \vartheta_s + \bu_S\cdot\nabla \vartheta_s  + su_T\right) 
d^2x_sdt
    \end{split}\label{eqn: CLagVar1}
\end{align}
where for convenience and brevity in notation we define
\begin{align}
B := \frac{\delta \ell}{\delta D } 
= \tfrac12 |\bu_S|^2 + \tfrac12 u_T^2 +  u_Tfx + {D}\gamma(z) \vartheta_s - p
\quad\hbox{and}\quad 
\gamma(z) := \frac{g}{\theta_0} \left(z - \frac{H}{2} \right)
\,,\label{def: gamma-B}
\end{align}
in which $B$ is the Bernoulli function.

We also have the continuity equation for the advected areal density ${D} d^2x_s$, written in its calculus form or in its Lie derivative form, respectively, as
\begin{align}
\p_t {D} + {\rm div}({D} \bu_S) = 0
\,,\quad\hbox{or}\quad 
\left(\p_t + {\cal L}_{u_s}\right)\left({D} d^2x_s\right) = 0
\,,\label{eqn: D}
\end{align}
in which ${\cal L}_{u_s}$ is the Lie-derivative with respect to the area-preserving vector field $u_s = \bu_S\cdot\nabla$, with time-dependent planar vector components $\bu_S(x,z,t)\in \mbb{R}^2$.  

Upon substituting the Euler-Poincar\'e variations into the variational results for $\delta S$ above in \eqref{eqn: CLagVar1} and integrating by parts following \cite{HMR1998}, one finds the equation of motion and auxiliary equations in Lie derivative form, as
\begin{align}
\begin{split}
\left(\p_t + {\cal L}_{u_s}\right)\left({D}^{-1}\bM\cdot d\bx \right) &= dB
\,,\\
\left(\p_t + {\cal L}_{u_s}\right)\left({D} d^2x_s\right) &= 0
\,,\\
\left(\p_t + {\cal L}_{u_s}\right)\left(\pi_t d^2x_s\right) &= -  {D}\gamma(z) d^2x_s
\,,\\
\left(\p_t + {\cal L}_{u_s}\right)\vartheta_s &= -s u_T
\,.
\end{split}
\label{eqns: Clebsch0}
\end{align}
Here, one defines the following momentum variables whose value will be constrained to $D=1$ 
by the Lagrange multiplier $p$ in \eqref{eqn: CLagVar1},
\begin{align}
\begin{split}
 \bM &:= {D}\bu_S - \pi_T \nabla \vartheta_s
 \,,\\
 \pi_T &:= s^{-1}{D} (u_T + fx) 
\,.\end{split}
\label{def: momvars}
\end{align}
\begin{theorem}[Kelvin-Noether theorem for the Euler-Boussinesq VSM]
	The Euler-Boussinesq vertical slice model in equation \eqref{eqns: Clebsch1} satisfies
	\begin{equation}
	\frac{d}{dt} \oint_{\gamma_t} 
     {D}^{-1}\bM\cdot d\bx
    = \oint_{\gamma_t}  dB = 0 \,,
	\end{equation}
	where $\gamma_t:C^1 \mapsto M$ is a closed loop moving with the flow $\gamma_t=\phi_t\gamma_0$
    generated by the vector field $u_s=\dot{\phi}_t\phi_t^{-1}$.
\end{theorem}

\begin{remark}[PV conservation for the VSM]\label{rmk: PV VSM}
Because $d({D}^{-1}\bM\cdot d\bx)=\yh\cdot{\rm curl}({D}^{-1}\bM) \,d^2x_s$, $d^2B=0$ and ${D}=1$, the first equation in \eqref{eqns: Clebsch0} implies potential vorticity advection (conservation of PV on fluid parcels)
\begin{align}
\left(\p_t + {\cal L}_{u_s}\right)d\left({D}^{-1}\bM\cdot d\bx \right) = 0
\quad\Longrightarrow\quad 
(\p_t+\bu_S\cdot\nabla)\left(\yh\cdot{\rm curl}\bu_S + J(\pi_T,\vartheta_s)\right) = 0
\,,\label{eqn: PVcalc0}
\end{align}
where $J(a,b)dx\wedge dz = da\wedge db$ defines the Jacobian operation for functions $(a,b)$ of $(x,z)\in \cal S$. If we also define $\bu_S:=\nabla^\perp\psi_s$ for a stream function $\psi_s$ since ${D}=1$ implies ${\rm div}\bu_S=0$, then $\yh\cdot{\rm curl}\bu_S=\Delta\psi_s$ for the Laplacian operator $\Delta$ on the vertical $(x,z)$ slice and one may write PV conservation on fluid parcels for the VSM as 
\begin{align}
\left(\p_t + \bu_S\cdot\nabla \right)q = \p_t q+ J(\psi_s,q)= 0
\quad\hbox{with}\quad 
q:= \Delta\psi_s + J(\pi_T,\vartheta_s)
\,.\label{eqn: PVcons0}
\end{align}
\end{remark}

\subsection{Hamiltonian formulation of VSMs}\label{sec: Ham VSM}
Recall that the Lagrangian for the Euler-Boussinesq Eady model is given by
\begin{equation}\tag{\eqref{eqn:Eady_action} revisited}
    \ell[u_S,u_T,{D},\vartheta_s,p] = \int_M \frac{{D}}{2}(|\bu_S|^2 + u_T^2) + {D}fu_Tx + \frac{g}{\vartheta_0}{D}\left( z - \frac{H}{2} \right) \vartheta_s + p(1-{D})\, d^2x \,.
\end{equation}
Following the approach taken in a previous variational derivation of the model \cite{CH2013}, we may define momentum variables with respect to both $u_S$ and $u_T$ as
\begin{equation}
    m_S = \frac{\delta\ell}{\delta u_S}\,,\quad\hbox{and}\quad m_T = \frac{\delta\ell}{\delta u_T} = s\pi_T \,,
\end{equation}
where $\pi_T$ is the momentum density introduced in \eqref{eqn: Cleb Lin VSM} whose relationship with $m_T$ will be discussed later in this section. Continuing with the $m_T$ definition of momentum density for the transverse velocity $u_T$, for the Euler-Boussinesq Eady system, we have
\begin{equation}
    m_S = \bm_S\cdot d\bx \otimes d^2x = {D}\bu_S\cdot d\bx \otimes d^2x \,,\quad\hbox{and}\quad m_T = {D}u_T + Dfx \,.
\end{equation}
Defining $\gamma(z) := \frac{g}{\vartheta_0} (z - \frac{H}{2})$, the Hamiltonian may be calculated as
\begin{align}
\begin{split}
    h[m_S,m_T,{D},\vartheta_s] &= \int_M \bm_S\cdot\bu_S + m_Tu_T \,d^2x - \ell[u_S,u_T,{D},\vartheta_s,p]
    \\
    &= \int_M \frac{|\bm_S|^2}{{D}} + m_T\left( \frac{m_T}{{D}}-fx \right) - \frac{{D}}{2}\bigg(\left|\frac{\bm_S}{{D}}\right|^2 + \left(\frac{m_T}{{D}} - fx \right)^2\bigg) 
    \\
    &\qquad\qquad\qquad - {D}fx\left(\frac{m_T}{{D}} - fx\right) - {D}\gamma(z)\vartheta_s + p({D}-1)\,d^2x
    \\
    &= \int_M \frac{|\bm_S|^2}{2{D}} + \frac{m_T}{2{D}} - m_Tfx + \frac{{D}}{2}(fx)^2 - {D}\gamma(z)\vartheta_s + p(D-1) \,d^2x
    \\
    &= \int_M \frac{|\bm_S|^2}{2{D}} + \frac{(m_T - {D}fx)^2}{2{D}} - {D}\gamma(z)\vartheta_s + p(D-1) \,d^2x \,.
\end{split}\label{eqn: tangled EB VSM ham}
\end{align}
It can be easily verified that this agrees with the conserved energy known to the literature \cite{CH2013, VCC2014}
\begin{equation}
    E = \int_M \frac{{D}}{2}\left( |\bu_S|^2 + u_T^2 \right) - {D}\gamma(z)\vartheta_s \,d^2x \,.
\end{equation}
With respect to the Hamiltonian, $h$, defined above, the equations of motion can be written in the following Lie-Poisson form

\begin{equation}
\frac{\p}{\p t}
\begin{bmatrix}\,
m_S \\ {D} \\ m_T \\ \vartheta_s
\end{bmatrix}
= - 
   \begin{bmatrix}
   \ad^*_{\Box}m_S & \Box \diamond {D} & \Box \diamond m_T  & \Box \diamond \vartheta_s
   \\
   \mathcal{L}_{\Box}{D} & 0 & 0 & 0 
   \\
   \mathcal{L}_{\Box}m_T & 0 & 0 & -s 
   \\
   \mathcal{L}_{\Box}\vartheta_s & 0 & s & 0
   \end{bmatrix}
   \begin{bmatrix}
	{\delta h}/{\delta m_S} \\
	{\delta h}/{\delta {D}} \\
	{\delta h}/{\delta m_T} \\
	{\delta h}/{\delta \vartheta_s}
\end{bmatrix} 
.
\label{Eqn: Tangled-Deterministic}
\end{equation}
\paragraph{Interpretations of the Poisson matrix \eqref{Eqn: Tangled-Deterministic}.}
The Poisson matrix appearing in \eqref{Eqn: Tangled-Deterministic} is ``tangled'' in the sense that it consists of the canonical Lie-Poisson structure on the dual of the semi-direct product Lie algebra $\mathfrak{s} = \mathfrak{X}(M)\ltimes \left(\mathcal{F}(M) \oplus \mathcal{F}(M) \oplus \text{Den}(M)\right)$ coupled to a $s$ weighted symplectic structure on the cotangent bundle $T^*\mathcal{F}(M) \simeq \mathcal{F}(M)\otimes \text{Den}(M)$. The $s$ factor on the canonical symplectic structure is due to the choice of momentum density $m_T$. Choosing the canonical momentum $\pi_T = m_T/s$, one can transform the weighted symplectic structure to the canonical structure as demonstrated below when `untangeling' the Hamiltonian structure.

In the following calculations, we will seek to verify that the equations are indeed Lie-Poisson equations. We first notice that equations \eqref{Eqn: Tangled-Deterministic} imply that, for any function $f$ on the semi-direct product Lie co-algebra, $\mathfrak{X}^* \ltimes (\Lambda^2 \oplus \Lambda^2 \oplus \Lambda^0 )$, we have
	\begin{align*}
		\frac{\p f}{\p t}(m_S,{D},m_T,\vartheta_s) &= \bigg\langle \bigg( \frac{\delta f}{\delta m_S} , \frac{\delta f}{\delta D}, \frac{\delta f}{\delta m_T}, \frac{\delta f}{\delta \vartheta_s} \bigg) , \bigg( \frac{\p m_S}{\p t} , \frac{\p D}{\p t} ,\frac{\p m_T}{\p t} ,\frac{\p\vartheta_s}{\p t}  \bigg) \bigg\rangle
		\\ 
		&=-\bigg\langle\bigg( \frac{\delta f}{\delta m_S} , \frac{\delta f}{\delta D}, \frac{\delta f}{\delta m_T}, \frac{\delta f}{\delta \vartheta_s} \bigg) , \bigg( \ad^*_{{\delta h}/{
		\delta m_S}}m_S + \frac{\delta h}{\delta D}\diamond D + \frac{\delta h}{\delta m_T}\diamond m_T + \frac{\delta h}{\delta\vartheta_s}\diamond \vartheta_s ,
		\\
		&\hspace{+5cm} \mathcal{L}_{{\delta h}/{\delta m_S}}D ,   \mathcal{L}_{{\delta h}/{\delta m_S}} m_T - s\frac{\delta h}{\delta \vartheta_s} ,   \mathcal{L}_{{\delta h}/{\delta m_S}}\theta_s + s\frac{\delta h}{\delta m_T}\bigg)\bigg\rangle
		\\
		&=-\bigg\langle m_S , \ad_{\delta h / \delta m_s}\frac{\delta f}{\delta m_S} \bigg\rangle - \bigg\langle  m_T, \mathcal{L}^T_{\delta h / \delta m_S}\frac{\delta f}{\delta m_T} - \mathcal{L}^T_{\delta f / \delta m_S}\frac{\delta h}{\delta m_T}  \bigg\rangle \\
		&\qquad - \bigg\langle \vartheta_s , \mathcal{L}^T_{\delta h / \delta m_S}\frac{\delta f}{\delta \vartheta_s} - \mathcal{L}^T_{\delta f / \delta m_S}\frac{\delta h}{\delta \vartheta_s} \bigg\rangle - \bigg\langle D , \mathcal{L}^T_{\delta h / \delta m_S}\frac{\delta f}{\delta D} - \mathcal{L}^T_{\delta f / \delta m_S}\frac{\delta h}{\delta D} \bigg\rangle \\
		&\qquad + s\bigg[ \bigg\langle \frac{\delta f}{\delta m_T},\frac{\delta h}{\delta \vartheta_s} \bigg\rangle - \bigg\langle \frac{\delta f}{\delta\vartheta_s},\frac{\delta h}{\delta m_T} \bigg\rangle \bigg] \,.
	\end{align*}
Noting that the Lie bracket for the semi-direct product algebra, $\mathfrak{X}\ltimes (\Lambda^0 \oplus \Lambda^0 \oplus \Lambda^2 ) $, is
	\begin{equation*}
		[(X, f,g,\omega),(\tilde X,\tilde f,\tilde g,\tilde\omega)] = ([X,\tilde X],X(\tilde f) - \tilde X(f) ,X(\tilde g) - \tilde X(g) ,X(\tilde \omega) - \tilde X(\omega) ) \,,
	\end{equation*}
	we thus have
    \begin{equation}\label{eqn:LP_demonstration_1}
	\begin{aligned}
		\frac{\p f}{\p t} &= \bigg\langle (m_S,{D},m_T,\theta_s), \bigg[\bigg(\frac{\delta f}{\delta m_S} , \frac{\delta f}{\delta D},\frac{\delta f}{\delta m_T},\frac{\delta f}{\delta \theta_s}\bigg),\bigg(\frac{\delta h}{\delta m_S} , \frac{\delta h}{\delta D},\frac{\delta h}{\delta m_T},\frac{\delta h}{\delta \theta_s}\bigg)\bigg] \bigg\rangle
		\\
		&\qquad +s\bigg[ \bigg\langle \frac{\delta f}{\delta m_T},\frac{\delta h}{\delta \theta_s} \bigg\rangle - \bigg\langle \frac{\delta f}{\delta\theta_s},\frac{\delta h}{\delta m_T} \bigg\rangle \bigg] \,.
	\end{aligned}
    \end{equation}
The right hand side of this equation is the conventional Lie-Poisson bracket on $\mathfrak{X}\ltimes (\Lambda^0 \oplus \Lambda^0 \oplus \Lambda^2 )$, and the remaining term can be shown to be a compatible bracket. Indeed, as defined in \cite{CMR2004}, for a bilinear skew-symmetric map, $\Sigma:\mathfrak{g}\times \mathfrak{g}\mapsto \mathbb{R}$, we may define the $\Sigma$-bracket on $\mathcal{F}(\mathfrak{g}^*)$ by $\{ g_1 , g_2 \}_{\Sigma}(\mu) = \Sigma\left( {\delta g_1}/{\delta \mu},{\delta g_2}/{\delta \mu} \right)$ for $\mu\in \mathfrak{g}^*$ and any $g_1,g_2 \in \mathcal{F}(\mathfrak{g}^*)$. As proved in \cite{CMR2004}, the sum of the Lie-Poisson bracket and the $\Sigma$-bracket, $\{\cdot,\cdot\}^\Sigma = \{\cdot,\cdot\} + \{\cdot,\cdot\}_\Sigma$, is itself a Poisson bracket on $\mathfrak{g}^*$ if and only if $\Sigma$ is a 2-cocycle. The final term in equation \eqref{eqn:LP_demonstration_1} is a $\Sigma$-bracket correspoinding to the bilinear map $\Sigma((X_1, f_1,g_1,\omega_1),(X_2, f_2,g_2,\omega_2)) = s\scp{g_1}{\omega_2} - s\scp{g_2}{\omega_1}$, and it can be easily verified that this bilinear map satisfies the 2-cocycle identity.

Thus, the equations \eqref{Eqn: Tangled-Deterministic} can be written in the following form
\begin{equation}\label{eqn:LP_demonstration_conclusion}
    \p_t f = \{ f , h \}^\Sigma \,,
\end{equation}
and have a Poisson structure.

\paragraph{`Untangled' Hamiltonian form.}
As was discussed in a previous work in the context of reduction by stages \cite{HHS2023c}, there exists a so-called \emph{untangling} map which allows us to untangle the Lie-Poisson structure by shifting the momentum variable. Indeed, consider the terms corresponding to the Legendre transformation in the previous section. Indeed, the sum of the Lagrangian and the Hamiltonian is
\begin{equation*}
    \ell + h = \scp{\bm_S}{\bu_S} + \scp{m_T}{u_T} \,.
\end{equation*}
By considering the advection equation for $\vartheta_s$, we have that
\begin{align}\label{eqn:untangled_momentum}
\begin{split}
    \ell + h &= \scp{\bm_S}{\bu_S} + \scp{m_T}{-\frac{\p_t\vartheta_s + \bu_S\cdot\nabla\vartheta_s}{s}}
    \\
    &= \scp{\bm_S - \frac{m_T\nabla\vartheta_s}{s}}{\bu_S} - \scp{\frac{m_T}{s}}{\p_t\vartheta_s} =: \scp{\bM}{\bu_S} - \scp{\pi_T}{\p_t \vartheta_s} \,,    
\end{split}
\end{align}
where the pair of momentum variables, $\bM$ and $\pi_T$ have the same definition as in \eqref{def: momvars}. Furthermore, the velocity variables, $\bu_S$ and $u_T$, can be written in terms of these momenta as
\begin{equation*}
    \bu_S = \frac{\bM + \pi_T\nabla\vartheta_s}{{D}} \,,\quad\hbox{and}\quad u_T = \frac{s\pi_T}{{D}} - fx \,,
\end{equation*}
and the Hamiltonian may be derived as follows.
\begin{equation}\label{eqn:untangled_hamiltonian}
\begin{aligned}
    \tilde{h}[\bM,\pi_T,{D},\vartheta_s] &= \int_M \bM \cdot\bu_S - \pi_T\p_t\vartheta_s - \ell[u_S,u_T,{D},\vartheta_s,p] \,d^2x
    \\
    &= \int_M \bM \cdot\bu_S + \pi_T(\bu_S\cdot\nabla\vartheta_s + su_T) - \ell[u_S,u_T,{D},\vartheta_s,p] \,d^2x
    \\
    &= \int_M \bM\cdot\left( \frac{\bM + \pi_T\nabla\vartheta_s}{{D}} \right) + \pi_T\left(\bigg(\frac{\bM + \pi_T\nabla\vartheta_s}{{D}}\bigg)\cdot\nabla\vartheta_s + s\bigg(\frac{s\pi_T}{{D}} - fx\bigg)\right) - \ell \,d^2x
    \\
    &= \int_M \frac{|\bM|^2}{{D}} + \frac{2\pi_T\bM\cdot\nabla\vartheta_s}{{D}} + \frac{\pi_T^2|\nabla\vartheta_s|^2}{{D}} + \frac{s^2\pi_T^2}{{D}} - sfx\pi_T - \frac{{D}}{2}\bigg| \frac{\bM + \pi_T\nabla\vartheta_s}{D} \bigg|^2
    \\
    &\qquad\qquad  - \frac{{D}}{2}\bigg( \frac{s\pi_T}{{D}} - fx \bigg)^2 - {D}fx\bigg( \frac{s\pi_T}{{D}} - fx \bigg) - {D}\gamma(z)\vartheta_s + p({D}-1)\,d^2x
    \\
    &= \int_M \frac{|\bM + \pi_T\nabla\vartheta_s|^2}{2{D}} + \frac{(s\pi_T - {D}fx)^2}{2{D}} - {D}\gamma(z)\vartheta_s + p({D}-1)\,d^2x \,.
\end{aligned}
\end{equation}

The equations of motion for the Hamiltonian defined in terms of ${\rm M}$ and $\pi_T$ have the following \emph{untangled} Lie-Poisson structure
\begin{equation}
\frac{\p}{\p t}
\begin{bmatrix}\,
{\rm M} \\ {D} \\ \pi_T \\ \vartheta_s
\end{bmatrix}
= - 
   \begin{bmatrix}
   \ad^*_{\Box}{\rm M} & \Box \diamond {D} & 0  & 0
   \\
   \mathcal{L}_{\Box}{D} & 0 & 0 & 0 
   \\
   0 & 0 & 0 & -1 
   \\
   0 & 0 & 1 & 0
   \end{bmatrix}
   \begin{bmatrix}
	{\delta\tilde h}/{\delta {\rm M}} \\
	{\delta\tilde h}/{\delta {D}} \\
	{\delta\tilde h}/{\delta \pi_T} \\
	{\delta\tilde h}/{\delta \vartheta_s}
\end{bmatrix} 
.\label{Eqn: Untangled-Deterministic}
\end{equation}
The Poisson structure appearing in equation \eqref{Eqn: Untangled-Deterministic} is untangled and it is the canonical Poisson structure on the dual space $\left(\mathfrak{X}^*(M) \ltimes \text{Den}(M)\right) \oplus T^*\mathcal{F}(M)$. That is, it is the sum of the Lie-Poisson matrix on the dual of the semi-direct product Lie algebra $\mathfrak{X}(M)\ltimes \mathcal{F}(M)$ and the canonical symplectic matrix on $T^*\mathcal{F}(M)$.
One can verify the equivalence of the tangled and untangled Hamiltonian form of the Lie-Poisson equations \eqref{Eqn: Tangled-Deterministic} and \eqref{Eqn: Untangled-Deterministic} respectively, by direct calculation. Indeed, the momentum equation in a result of the relationship
\begin{equation*}
    (\p_t + \mathcal{L}_{u_S})(\pi_Td\vartheta_s) = -m_T d\frac{\delta h}{\delta m_T} + \frac{\delta h}{\delta \vartheta_s}d\vartheta_s = -\frac{\delta h}{\delta\vartheta_s}\diamond \vartheta_s - \frac{\delta h}{\delta m_T}\diamond m_T \,,
\end{equation*}
and it can be easily verified that the above Lie-Poisson system is equivalent to the Euler-Boussinesq Eady equations \eqref{eqn:EBE-slice}-\eqref{eqn:EBE-incompressibility} for the Hamiltonian, $\tilde h$, given by equation \eqref{eqn:untangled_hamiltonian}.

\subsection{Computational Simulations of VSMs}\label{sec: CompSims}
This section demonstrates some of the variety of solution behaviours of the VSM derived in the previous sections by showing the dynamics of front formation in computational simulations of the VSM equations.

For this section only, we rewrite the VSM equations \eqref{eqn:EBE-slice} - \eqref{eqn:EBE-incompressibility} using the buoyancy variable $b_s(x,z,t)$ which is related to the potential temperature $\vartheta_s$ by the relation $b_s = \frac{g}{\vartheta_0} \vartheta_s$. We assume the buoyancy decomposes to $b_s(x,z,t) = \overline{b}_s(z) + b'_s(x,z,t)$ where $\overline{b}_s(z)$ is the background buoyancy which is linear in $z$. Upon defining the buoyancy frequency as $N^2 = {\p \overline{b}_s/\p z}$, the VSM equations can be written as
\begin{align}
    \begin{split}
    \p_t \bu_S + \bu_S \cdot \nabla \bu_S &= fu_T\wh{x} + b'_s\wh{z} - \nabla p \,,\\
	\p_t u_T + \bu_S \cdot \nabla u_T  &= - f\bu_S\cdot\wh{x}
    -s\frac{g}{\vartheta_0}\left( z - \frac{H}{2} \right) \,, \\
	\p_t b'_s + \bu_S \cdot \nabla b'_s + N^2 \bu_S\cdot \wh{z} + s\frac{g}{\vartheta_0} u_T &= 0 \,,\\
	\nabla \cdot \bu_S &= 0 \,.
    \end{split}\label{eqn:EBE buoyancy}
\end{align}
The numerical method we employ here to solve the VSM equations \eqref{eqn:EBE buoyancy} is the compatible finite element method \cite{CS2012} based on the finite element discretisation for VSMs discussed in \cite{YSCMC2017}. We summarise the numerical methods as follows.
The vertical slice domain $M$ is discretised on the mesh $\Omega$ using quadrilateral finite elements. For this discretisation, we approximate the prognostic variables $(\bu_S, u_T, b'_s, p)$, respectively, in terms of the finite element spaces $$(\mathring{RK_1(\Omega)}, DG_1(\Omega), CG_2(\Omega_h)\otimes DG_1(\Omega_v), DG_1(\Omega))\,.$$ Here, the various finite element function spaces are denoted as follows: $RK_1(\Omega)$ denotes the Raviart–Thomas space of polynomial degree $1$; $DG_1(\Omega)$ denotes the continuous finite element space of polynomial degree $1$; and $CG_2(\Omega)$ denotes the continuous finite element space of polynomial degree $2$. The space $\mathring{RK_1(\Omega)}$ is a subspace of $RK_1(\Omega)$ defined by $\mathring{RK_1(\Omega)} = \{\bu \in RK_1(\Omega) : \bu \cdot \mb{n} = 0 \quad\hbox{on}\quad\partial\Omega \}$. The finite element space for the buoyancy, $CG_2(\Omega_h)\otimes DG_1(\Omega_v)$, is the continuous finite element space in the horizontal direction and tensor product discontinuous finite element space in the vertical direction, so as to replicate Charney–Phillips grid staggering.

For the temporal discretisation, we use a semi-implicit scheme where we modify the explicit third order strong stability preserving Runge Kutta (SSPRK3) scheme by introducing an implicit time discretisation average applied to the forcing terms, as well as the advecting velocity field $\bu_S$ in the advection terms in all of the evolution equations. The resulting nonlinear system is solved using a Picard iterations scheme in which  the magnitude of the residual is enforced to be less than a prescribed tolerance.

The physical parameters used in the example simulation are obtained by adapting the parameters used in \cite{YSCMC2017} for oceanic conditions, rather than atmospheric. The computational domain is a rectangle $[-L, L]\times [0, H]$ where $L = 10000m$ and $H = 100m$. The boundary conditions are periodic on the lateral boundaries and the in-slice velocity $\bu_S$ has no normal component on the top and bottom boundaries. These boundary conditions are sufficient in the absence of explicit diffusion terms. The aspect ratio is given by $\sigma := H/L = 10^{-2}$ which is applicable to ocean submesoscale dynamics. The rotation frequency $f = 10^{-4}s^{-1}$, gravity $g=10ms^{-2}$, squared buoyancy frequency $N^2 := \frac{\p \overline{b}_s}{\p z} = \frac{g}{\vartheta_0}\frac{\p \vartheta_s}{\p z} = 10^{-6}s^{-2}$, the constant $y$ derivative of potential temperature $s = {\p\vartheta/\p y}$ and reference temperature $\vartheta_0$ are set such that ${gs/\vartheta_0} = 10^{-7}$. The initial conditions of $b'_s$ is taken to be a perturbation away from the stable, stratified buoyancy field which takes the form as in \cite{YSCMC2017},
\begin{align} 
    b'_s(x,z,0) = aN\left(-\left[1-\frac{1}{4}\coth\left(\frac{1}{4}\right)\right]\sinh Z \cos\left(\frac{\pi x}{L}\right) -  \frac{n}{4}\cosh Z \sin\left(\frac{\pi x}{L}\right)\right)\,, \label{eqn:init buoyancy}
\end{align}
where $a = 7.5$ is the amplitude of the perturbation and the constant $n$ and modified vertical coordinate $Z$ are defined as
\begin{align}
    n = 2\left(\left[\frac{1}{4}-\tanh\left(\frac{1}{4}\right)\right]\left[\frac{1}{4}-\coth\left(\frac{1}{4}\right)\right]\right)^{1/2}\!\!, \qquad Z = \frac{1}{2}\left(\frac{z}{H} - \frac{1}{2}\right)\,.
\end{align}
Three nondimensional parameters are involved. These are the Rossby number $Ro = U/fL$, Froude number $Fr = U/NH$ and Burger number $Bu= (NH)/(fL) = N\sigma/f = Ro/Fr$ which can be evaluated to $Bu = 0.1$. 

The initial pressure is calculated from hydrostatic balance condition using the initialised buoyancy field \eqref{eqn:init buoyancy} and the initial transverse velocity $u_T$ is calculated using the geostrophic balance condition using the computed initial pressure. Finally, the initial in-slice velocity $\bu_S$ is assumed to satisfy a semi-geostrophic condition which is calculated from the linearised evolution equation of $u_T$ and $b'_s$. For a full account of the finite element discretisation of the VSM, see \cite{YSCMC2017}.

Using the aforementioned model setup, we simulate the VSM for $30$ days with time-step $\delta t = 50s$. The snapshots of the numerical simulations of the VSM \eqref{eqn:EBE buoyancy} are shown in Figure \ref{fig:vsm snapshot v b} and Figure \ref{fig:vsm snapshot u p}. In Figure \ref{fig:vsm snapshot v b}, the left hand column and the right hand column are the snapshot of the transverse velocity $u_T$ and buoyancy $b'_s$ respectively, both of which are shown at day $4, 7,9$ and $14$ of the simulation.

The model displays signs of frontogenesis in the $u_T$ field at day $4$ when sharp gradients start consolidating near the center of the domain. At the same time, the buoyancy field is tracking the region of large gradients in $u_T$ and large buoyancy gradients can be found in the north east and south west regions. At day $7$, the front can be seen clearly in both the $u_T$ and $b'_s$ field and it is tilted eastwards starting from the bottom of the domain. The front weakens over time and is then regenerated with a westward tilt on day $9$, as is visible on the $u_T$ snapshot. The tilted front then rotates clockwise and displays a westward tilt on day $14$ whose magnitude is similar to the westward tilting front on day $7$. The weakening and reformation of the front continue over another $7$ day cycle. Consequently, we conclude that the VSM dynamics shows a quasi-periodic formation of fronts.

Figure \ref{fig:vsm snapshot u p} consists of the snapshot of the horizontal component of the velocity field $\bu_S$ and the pressure field $p$ at day $14$. The presence of the low pressure zone in the center of the pressure field aligned with the front in the $u_T$ field suggests a cyclone/anticyclone pair. This is further suggested by the lack of horizontal velocity in frontal regions. Thus, the Lagrangian fluid particles do not cross the front. Instead, they travel vertically upwards and downwards on the left and right of the frontal regions, respectively. The difference in sign of $u_T$ across the front then completes the interpretation of this configuration of the VSM simulation as a cyclone/anticyclone pair. 
\begin{figure}
    \centering
    \begin{subfigure}[b]{0.45\textwidth}
        \centering
        \includegraphics[width=\textwidth]{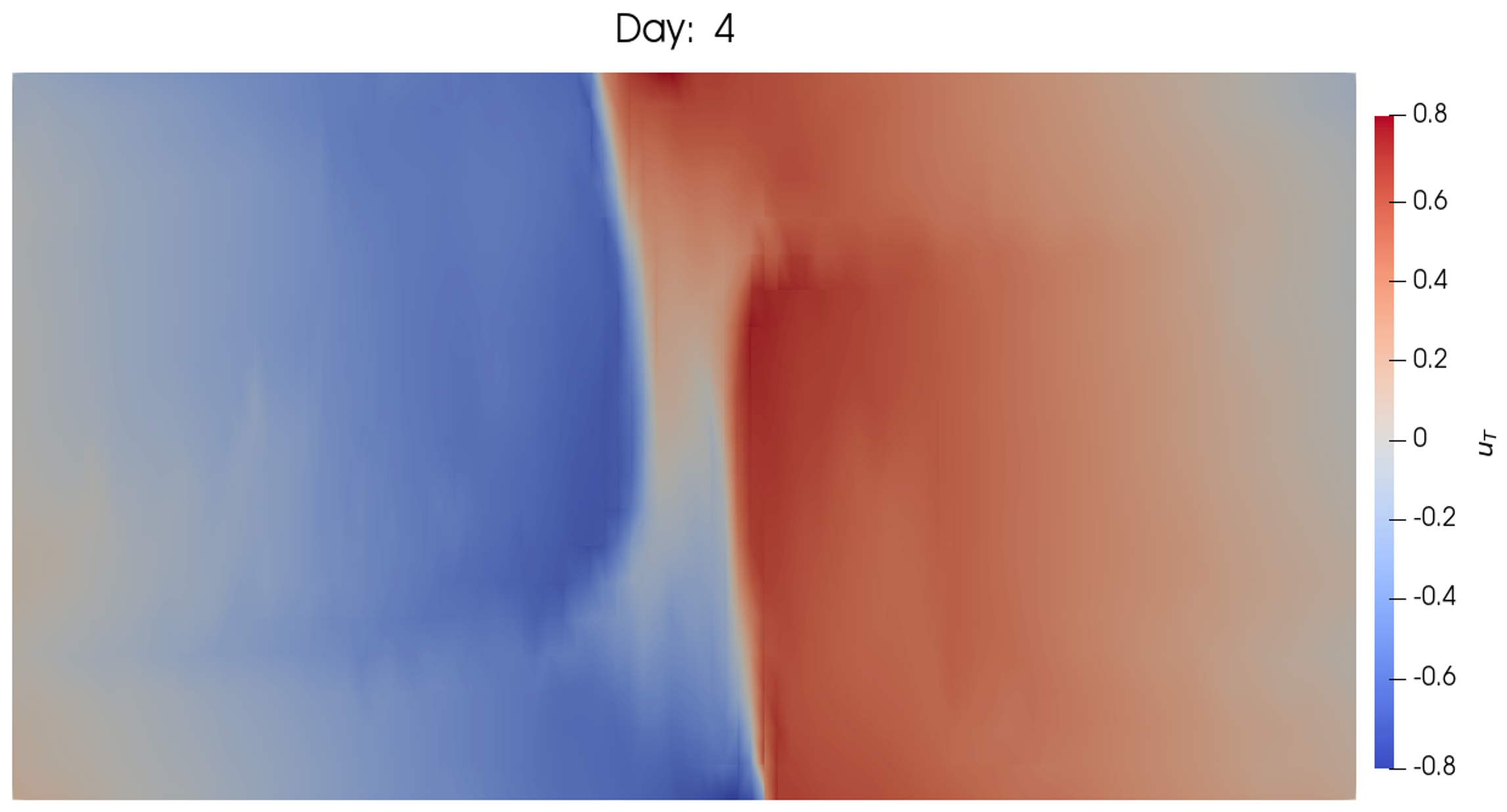}
    \end{subfigure}
    \begin{subfigure}[b]{0.45\textwidth}
        \centering
        \includegraphics[width=\textwidth]{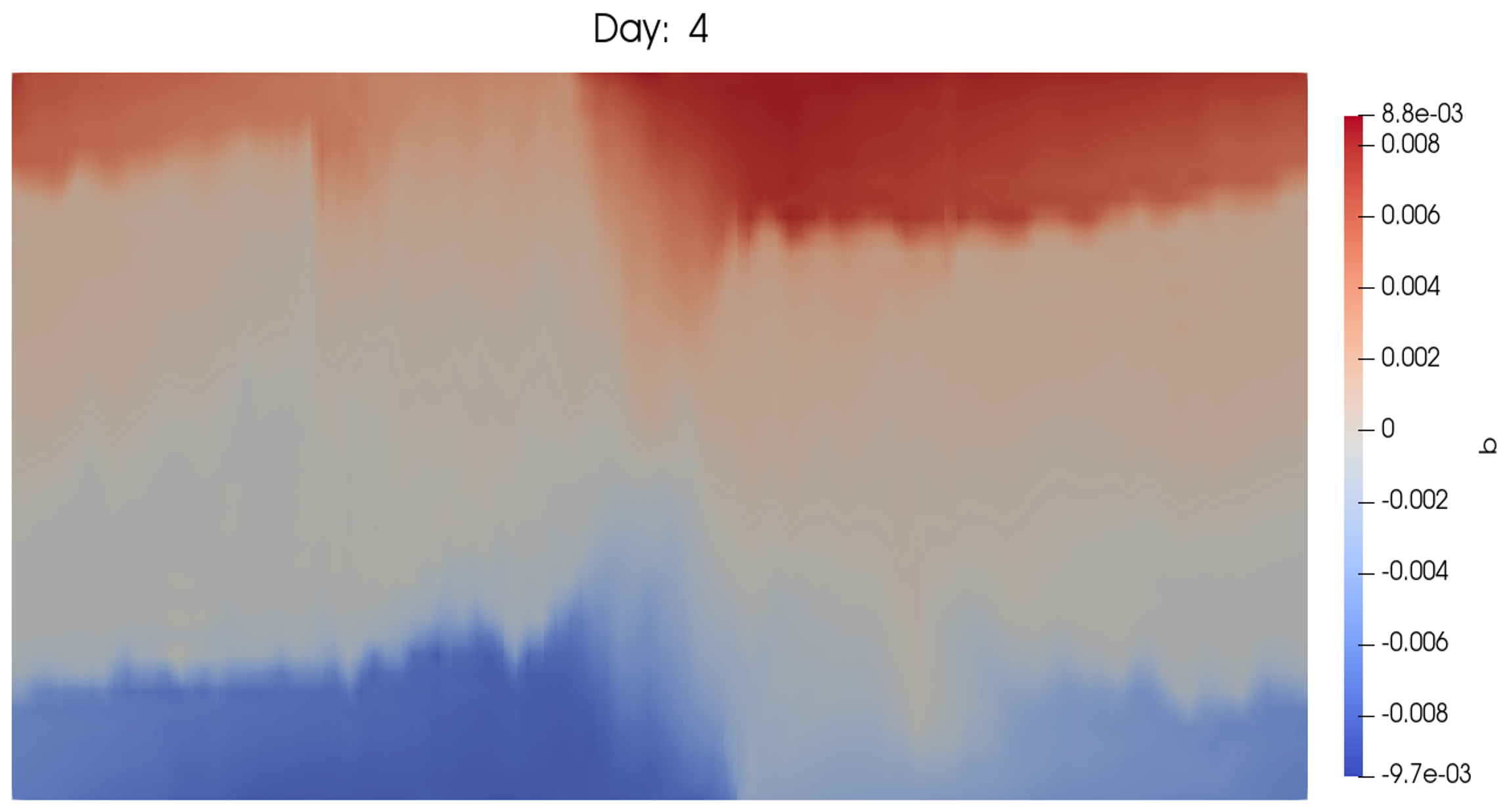}
    \end{subfigure}
    \begin{subfigure}[b]{0.45\textwidth}
        \centering
        \includegraphics[width=\textwidth]{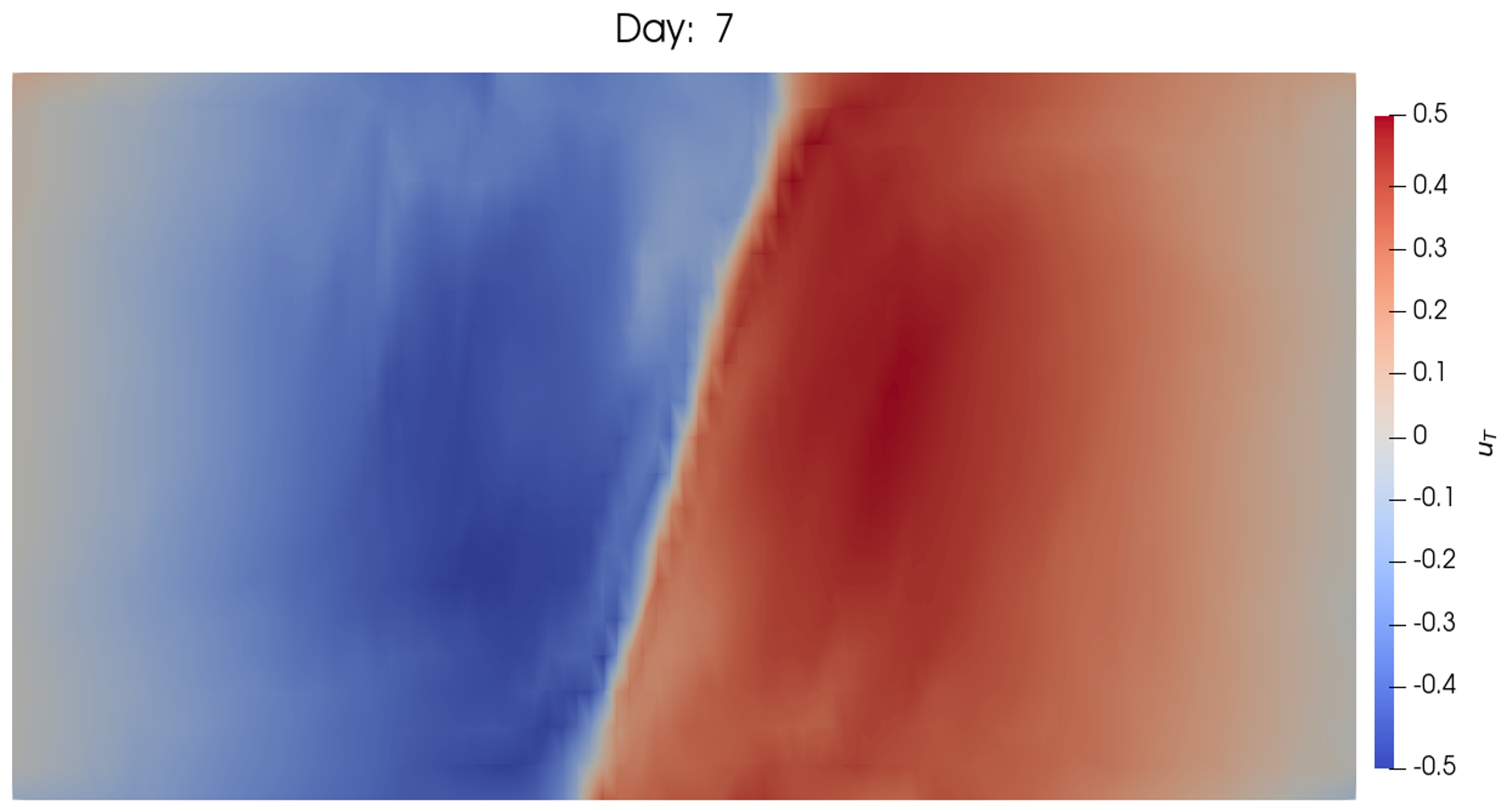}
    \end{subfigure}
    \begin{subfigure}[b]{0.45\textwidth}
        \centering
        \includegraphics[width=\textwidth]{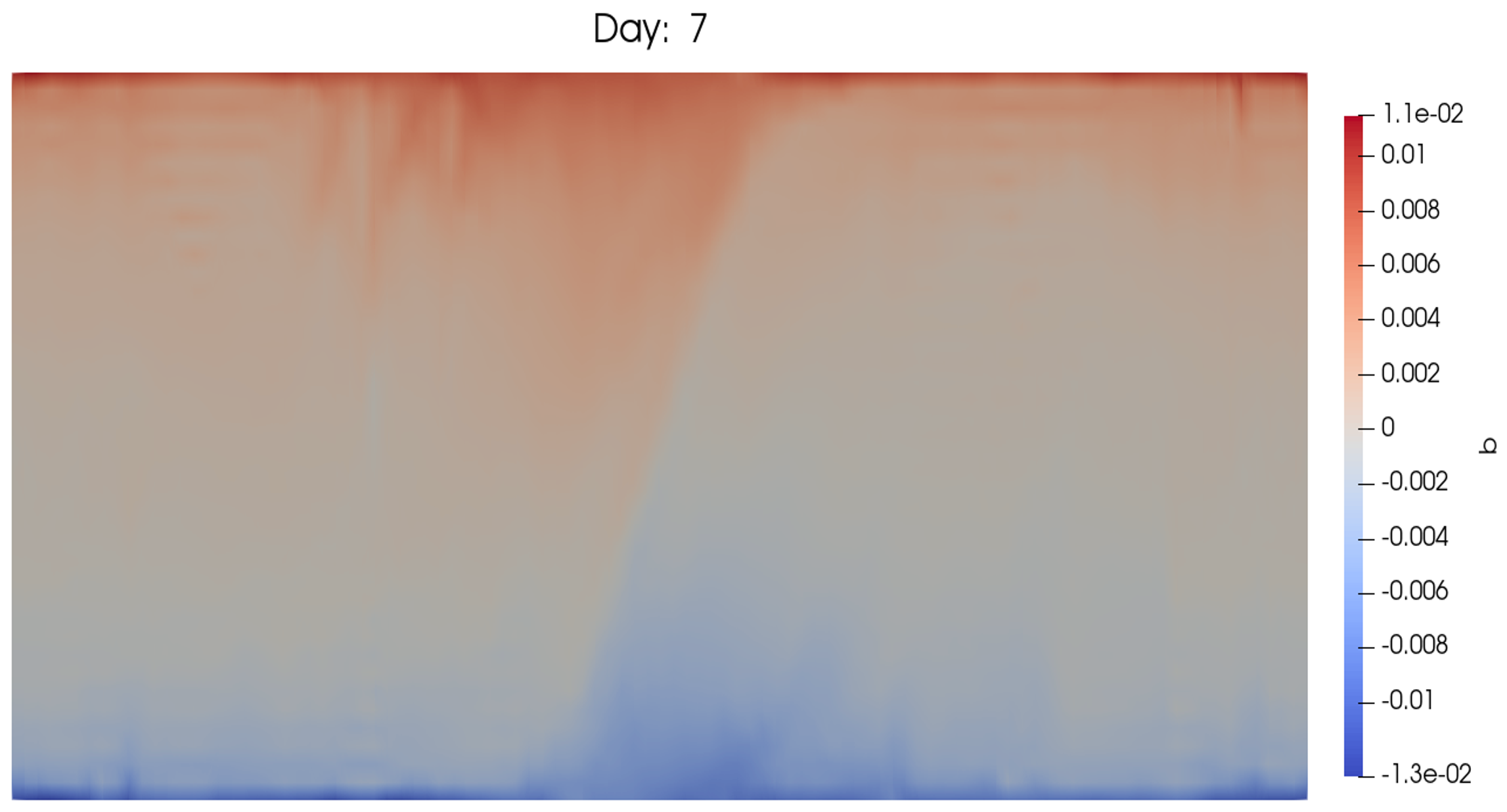}
    \end{subfigure}
        \begin{subfigure}[b]{0.45\textwidth}
        \centering
        \includegraphics[width=\textwidth]{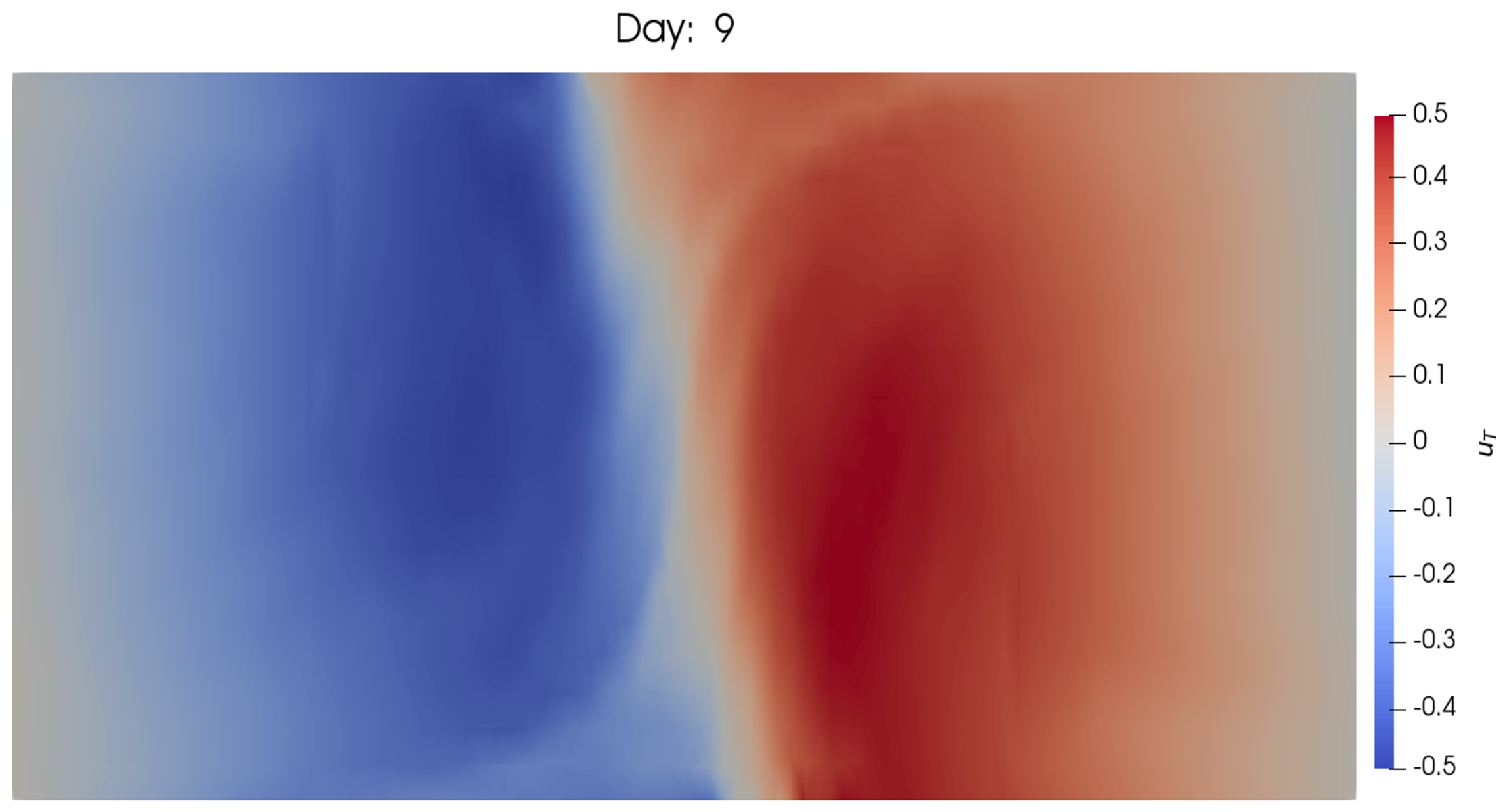}
    \end{subfigure}
    \begin{subfigure}[b]{0.45\textwidth}
        \centering
        \includegraphics[width=\textwidth]{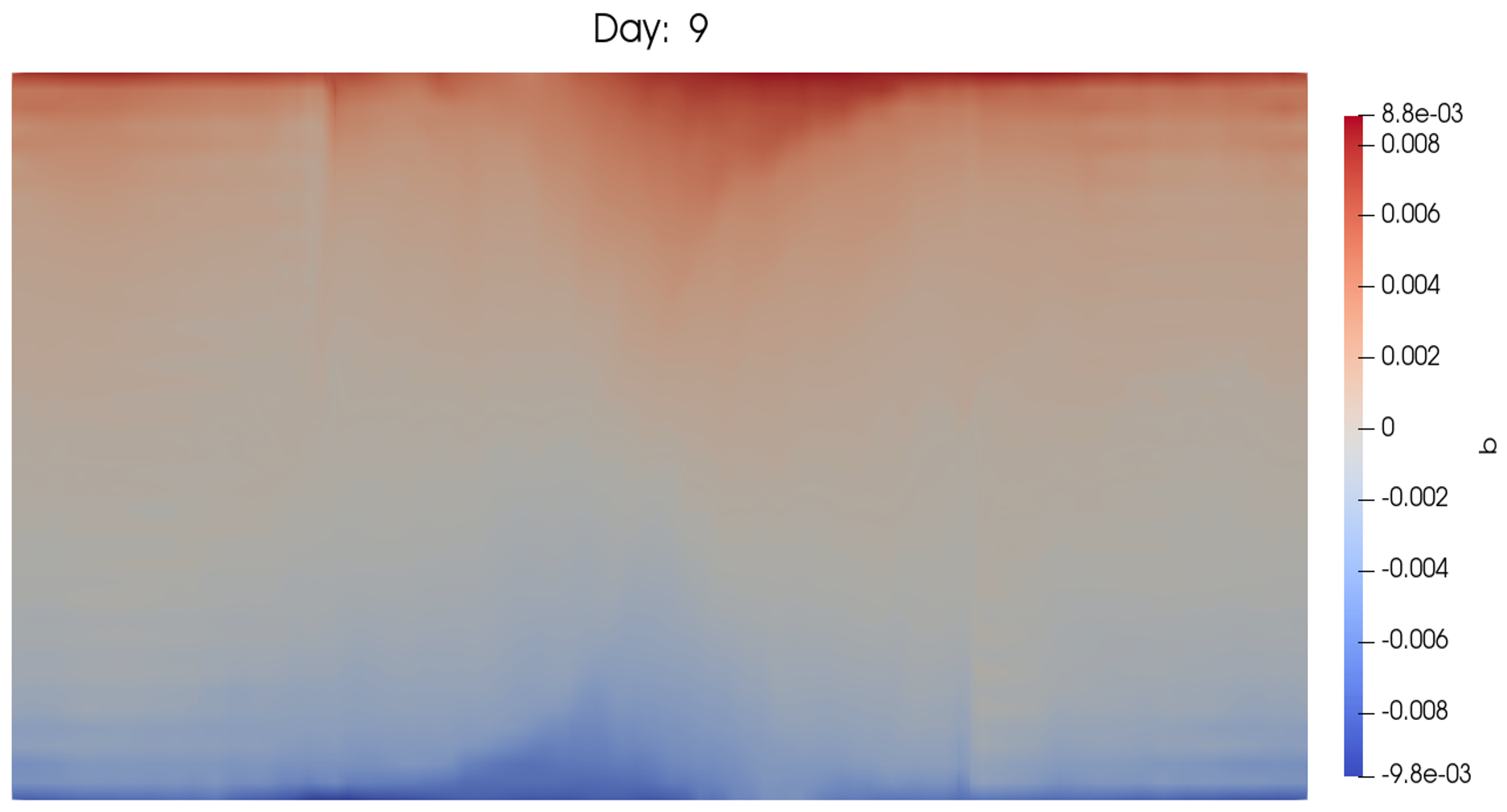}
    \end{subfigure}
        \begin{subfigure}[b]{0.45\textwidth}
        \centering
        \includegraphics[width=\textwidth]{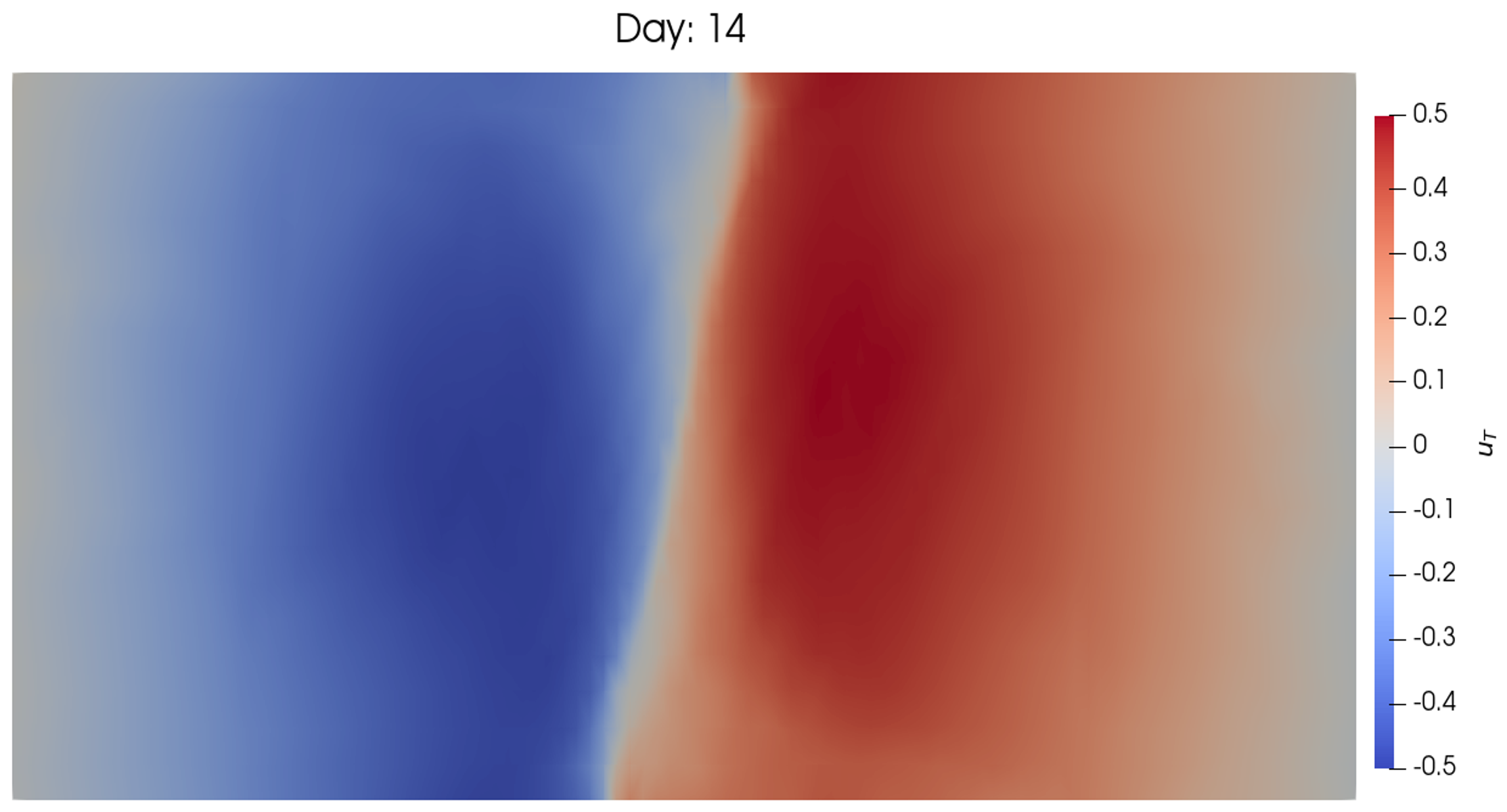}
    \end{subfigure}
    \begin{subfigure}[b]{0.45\textwidth}
        \centering
        \includegraphics[width=\textwidth]{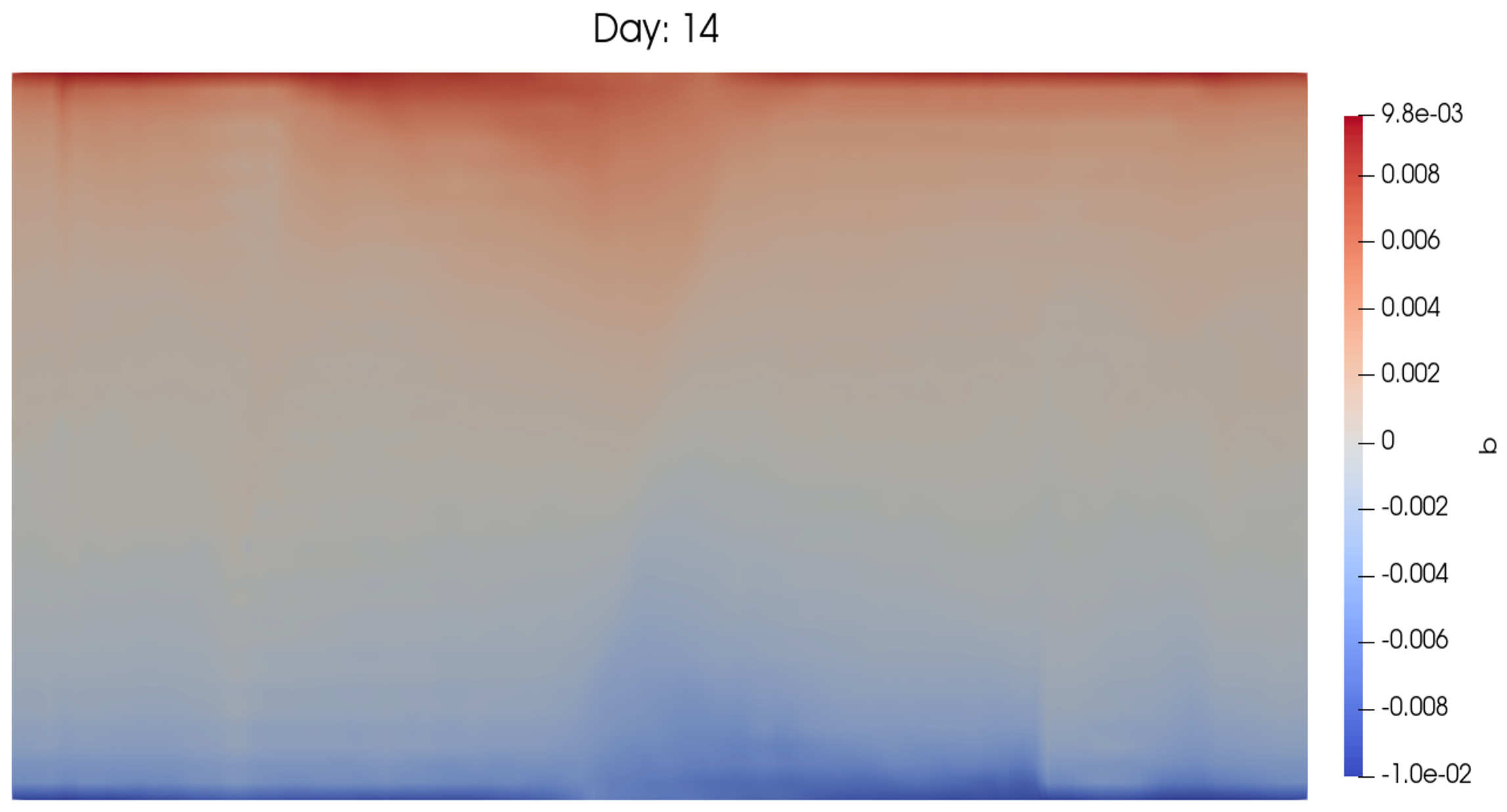}
    \end{subfigure}
    \caption{Snapshots of the transverse velocity field $u_T$ (left) and the buoyancy field $b'_s$ (right) of the numerical simulation of the VSM \eqref{eqn:EBE buoyancy} at days $4,7,9$ and $14$. The solutions are periodic on the lateral boundaries and their in-slice velocity $\bu_S$ has no normal component on the top and bottom boundaries; so the buoyancy field is advected horizontally at the top and bottom of the domain.}
    \label{fig:vsm snapshot v b}
\end{figure}

\begin{figure}
    \centering
    \begin{subfigure}[b]{0.45\textwidth}
        \centering
        \includegraphics[width=\textwidth]{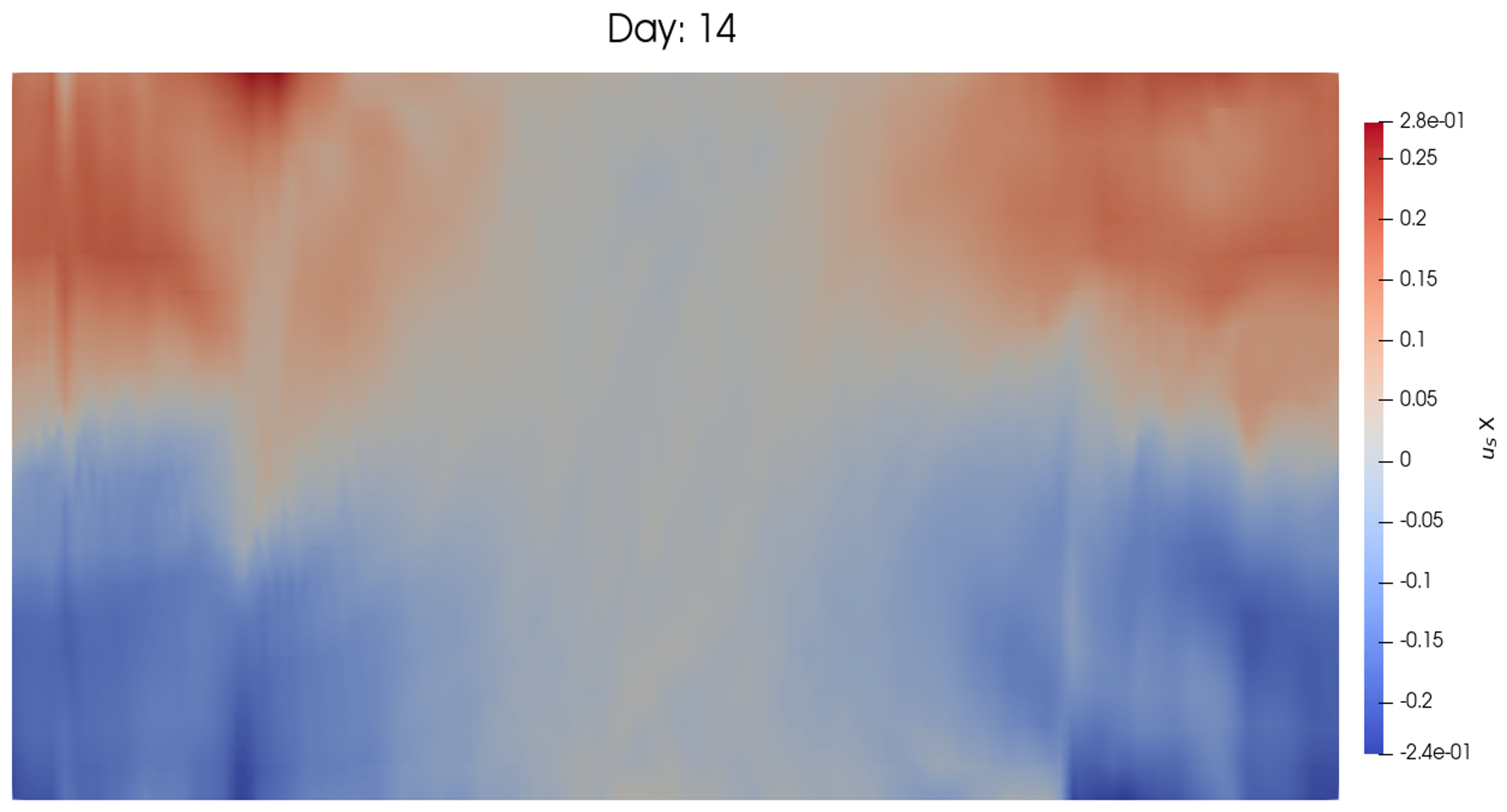}
    \end{subfigure}
    \begin{subfigure}[b]{0.45\textwidth}
        \centering
        \includegraphics[width=\textwidth]{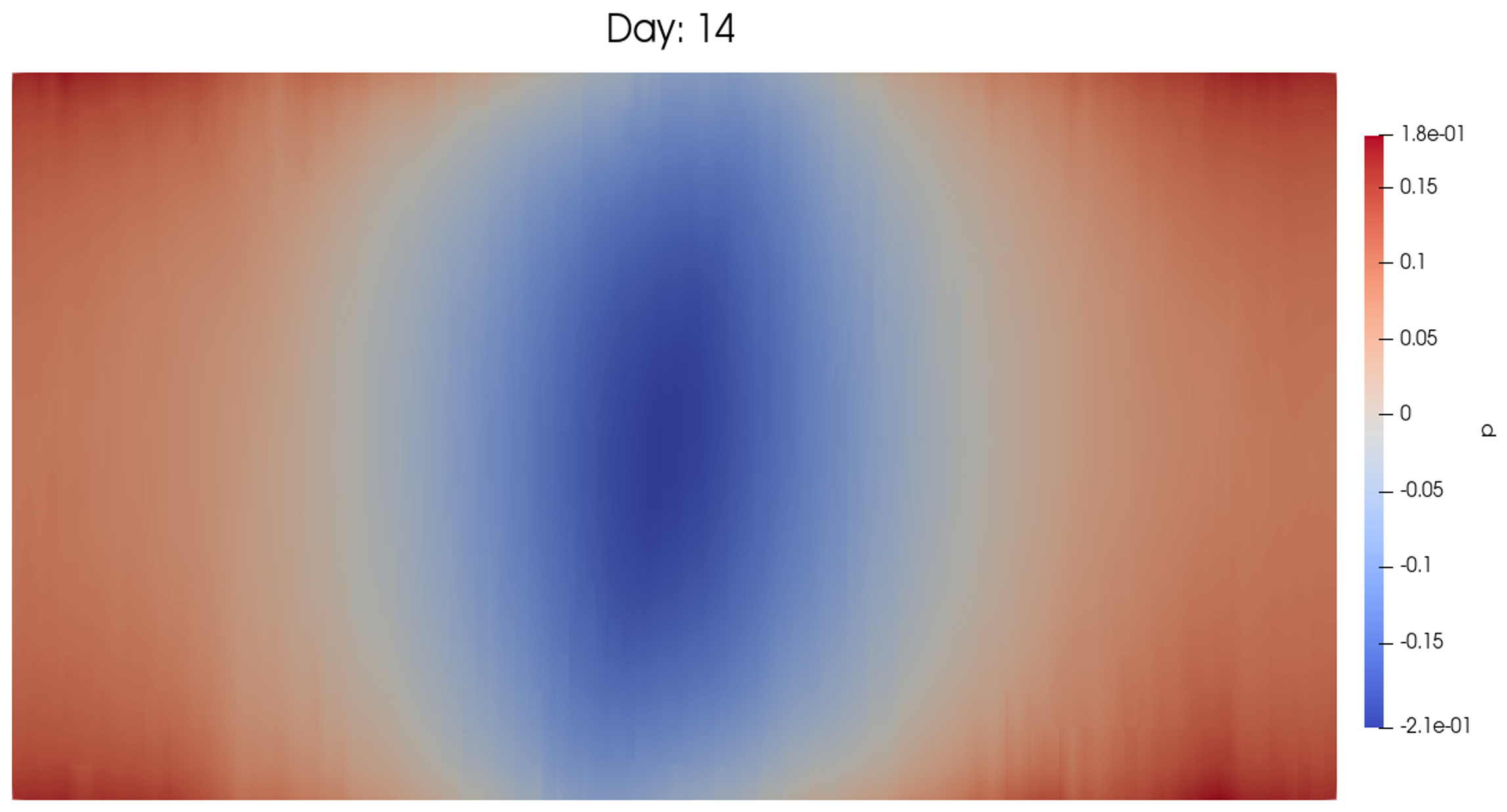}
    \end{subfigure}
    \caption{Snapshots of the horizontal component of in slice velocity field $\bu_S\cdot \wh{\bx}$ (left) and the pressure field $p$ (right) of the numerical simulation of the VSM \eqref{eqn:EBE buoyancy} at day $14$. }
    \label{fig:vsm snapshot u p}
\end{figure}


\section{Wave mean flow interaction (WMFI)}\label{sec:WMFI}

Describing the interacting dynamics between a mean flow and the fluctuating motion around this is a central problem in geophysical fluid dynamics. For fluid flows on a planetary scale, the average behaviour of the system is of particular interest since it can be computationally resolved. However, the interactions between the mean and fluctuating components of the flow are integral to understanding the formation and subsequent dynamics of instabilities. These instabilities can dramatically influence the nature of the mean flow and suitable parameterisations are required to model them effectively \cite{GM1990}. This section explores WMFI within the context of VSMs.

\subsection{Clebsch-Lin formulation of the EB VSM with wave dynamics}\label{sec:EB_VSM_IGW}
One general formulation of the coupling between the EB VSM dynamics and wave activity, let us first assume the general form of the wave dynamics is described by a wave Hamiltonian $H_W(N, \phi)$.
Here the wave degree of freedom $(N, \phi)$ are wave action density and wave phase respectively with the wave number $\bk = \nabla \phi$ which induces a single frequency wave propagating on the vertical slice domain $\cal S$ with area element $d^2x_s=dxdz$.
The coupling between the EB VSM degree of freedom and waves degree of freedom occurs in the Hamilton variational principle where the phase space principle of the wave motion are boosted to the frame of the in slice fluidic motion, in the same fashion as the out of slice fluid variables.
That is, we augment the Clebsch-Lin form of the Hamilton variational principle in equation \eqref{eqn: CLagVar1} to have the following
\begin{align}
\begin{split}
0 = \delta S &= \delta \int_a^b \int_{\cal S} 
\ell(\bu_S,u_T,D,\vartheta_s,\phi;p,\pi_t,N)
\\&=
\delta \int_a^b \int_{\cal S}\big(
\tfrac12 {D} |\bu_S|^2 + \tfrac12 {D}u_T^2 + {D} u_Tfx + {D}\gamma(z) \vartheta_s - p({D}-1)
\\&\hspace{2cm}
- \pi_T \left( \p_t \vartheta_s + \bu_S\cdot\nabla \vartheta_s  + su_T\big) 
- N(\partial_t\phi  + \bu_S\cdot\nabla \phi)\,d^3x + H_W(N,\bk) 
\right)dt
\\&=
 \int_a^b \int_{\cal S}
B \,\delta {D} 
+ \delta \bu_S \cdot \left( {D} \bu_S - \pi_T \nabla \vartheta_s - N \nabla \phi \right) 
+ \delta u_T  \left( {D} (u_T + fx) - s\pi_T \right) 
- \delta p ({D}-1)
\\&\hspace{15mm}
+ \delta \vartheta_s \left( \p_t \pi_T + {\rm div}(\pi_t \bu_S) + {D}\gamma(z) \right)
- \delta \pi_T \left( \p_t \vartheta_s + \bu_S\cdot\nabla \vartheta_s  + su_T\right) \
\\&\hspace{15mm}
+ \delta \phi \left( \partial_t N  + {\rm div}\left( N\bu_S + \frac{\delta H_W}{\delta \bk} \right)\right)
- \delta N \left( \p_t \phi + \bu_S\cdot\nabla \phi  - \frac{\delta H_W}{\delta N}\right) \
d^2x_sdt
\,.\end{split}
\label{eqn: CLagVar1wave}
\end{align}
Here, we have arbitrary variations for the two pairs of Clebsch variables $(\vartheta_s, \pi_T), (\phi, N)$ and the associated out of slice velocities $u_T$; we also have constrained variations for the Euler-Poincar\'e variables 
\begin{align}
    \delta u_s = \p_t\xi - {\rm ad}_{u_s}\xi \quad\text{and}\quad \delta ({D} d^2x_s)=-\,{\cal L}_{\xi}({D} d^2x_s)\,, \label{eqn: EP variations}
\end{align}
for arbitrary variations $\xi$. The constrained variations of $D d^2x_s$ implies the advection of $D d^2x$, written in coordinates as
\begin{align}
    \p_t {D} + {\rm div}({D} \bu_S) = 0
    \,,\quad\hbox{or}\quad 
    \left(\p_t + {\cal L}_{u_s}\right)\left({D} d^2x_s\right) = 0\,,
\end{align}
As before in \eqref{def: gamma-B}, we define the Bernoulli function, $B$, and $g$, $\theta_0$, $H$ are constants, as
\begin{align}
B := \frac{\delta \ell}{\delta D } 
= \tfrac12 |\bu_S|^2 + \tfrac12 u_T^2 +  u_Tfx + {D}\gamma(z) \vartheta_s - p
\quad\hbox{and}\quad 
\gamma(z) := \frac{g}{\theta_0} \left(z - \frac{H}{2} \right)
\,.\label{def: gamma-Bwave}
\end{align}
Let us briefly remark on the terms appearing in the Lagrangian \eqref{eqn: CLagVar1wave}. The first line is the fluid Lagrangian for EB fluids in the vertical slice \cite{HMR1998}. The first term in the second line is the phase-space Lagrangian that couples the thermal degree of freedom in the slide to the transverse fluid flow through the slice via the pairing $-\int_\mathcal{S} \pi_T\nabla \phi\cdot \bu_S\,d^2x$. Similarly, the second term and third term in the second line is the phase-space Lagrangian that couples the wave dynamics to the in slice fluid flow via the coupling $-\int_\mathcal{S} N\nabla \phi\cdot \bu_S\,d^2x$. These two coupling terms in the Lagrangian in \eqref{eqn: CLagVar1wave} introduces additional momentum components to the the fluid momentum in the slice. 

Upon substituting the Euler-Poincar\'e variations \eqref{eqn: EP variations} into the variational results for $\delta S$ above in \eqref{eqn: CLagVar1wave} and integrating by parts, one finds the equation of motion and auxiliary equations in Lie derivative form, as
\begin{align}
\begin{split}
\left(\p_t + {\cal L}_{u_s}\right)\left({D}^{-1}\wt{\bM}\cdot d\bx \right) &= dB
\,,\\
\left(\p_t + {\cal L}_{u_s}\right)\left(\pi_t d^2x_s\right) &= -  {D}\gamma(z) d^2x_s
\,,\\
\left(\p_t + {\cal L}_{u_s}\right)\vartheta_s &= -s u_T
\,,\\
\left(\p_t + {\cal L}_{u_s}\right)\left(N d^2x_s\right) &= - {\rm div}\left(\frac{\delta H_W}{\delta \bk}\right)\,d^2x_s\,,\\
\left(\p_t + {\cal L}_{u_s}\right)\phi &= \frac{\delta H_W}{\delta N}
\,,\\
\left(\p_t + {\cal L}_{u_s}\right)\left({D} d^2x_s\right) &= 0\,, \quad \text{with}\quad D=1\,.
\end{split}
\label{eqns: Clebsch1}
\end{align}
Here we have the incompressibility constraint $D=1$ enforced by the Lagrange multiplier $p$ and we have the following definitions for the momentum variables
\begin{align}
\begin{split}
 \wt{\bM} &:= {D}\bu_S - \pi_T \nabla \vartheta_s - N\bk
 \,,\\
 \pi_T &:= s^{-1}{D} (u_T + fx) 
\,.\end{split}
\label{def: momvars wave}
\end{align}
Compared with the standard EB VSM in \eqref{def: momvars}, we note that the untangled in slice fluid momentum $\wt{\bM}$ also have contributions from the additional wave degree of freedom. 
\begin{theorem}[Kelvin-Noether theorem for the Euler-Boussinesq VSM]
	The Euler-Boussinesq vertical slice model in equation \eqref{eqns: Clebsch1} satisfies
	\begin{equation}
	\frac{d}{dt} \oint_{\gamma_t} 
     {D}^{-1}\wt{\bM}\cdot d\bx
    = \oint_{\gamma_t}  dB = 0 \,,
    \label{def: KelThm Mtot}
	\end{equation}
	where $\gamma_t:C^1 \mapsto M$ is a closed loop moving with the flow $\gamma_t=\phi_t\gamma_0$
    generated by the vector field $u_s=\dot{\phi}_t\phi_t^{-1}$.
\end{theorem}

\begin{remark}[PV conservation for the VSM with waves]
Under similar considerations as for PV conservation for the EB VSM (without waves) in Remark \ref{rmk: PV VSM}, i.e. $d({D}^{-1}\wt{\bM}\cdot d\bx)= \yh\cdot{\rm curl}({D}^{-1}\wt{\bM}) \,d^2x_s$, $d^2B=0$ and ${D}=1$, the first equation in \eqref{eqns: Clebsch1} implies potential vorticity advection (conservation of PV on fluid parcels)
\begin{align}
\left(\p_t + {\cal L}_{u_s}\right)d\left({D}^{-1}\wt{\bM}\cdot d\bx \right) = 0
\quad\Longrightarrow\quad 
(\p_t+\bu_S\cdot\nabla)\left(\yh\cdot{\rm curl}\bu_S + J(\pi_T/D,\vartheta_s) + J(N/D,\phi) \right) = 0
\,,\label{eqn: PVcalc}
\end{align}
where, as before, $J(a,b)dx\wedge dz = da\wedge db$. We may write PV conservation on fluid parcels for the VSM with wave dynamics as 
\begin{align}
\left(\p_t + \bu_S\cdot\nabla \right)\wt{q} = \p_t \wt{q} + J(\psi_s,\wt{q})  = 0
\quad\hbox{with}\quad 
\wt{q}:= \Delta\psi_s + J(\pi_T/D,\vartheta_s) + J(N/D,\phi)
\,,\label{eqn: PVcons}
\end{align}
where $\psi_s$ is the stream function for the slice velocity in same manner as presented in Remark \ref{rmk: PV VSM}.
\end{remark}

Next, let's rewrite the Kelvin circulation theorem in equation \eqref{def: KelThm Mtot} in terms of the fluid momentum alone, 
\begin{align}
\bm_S = \wt{\bM} + \pi_T \nabla \vartheta_s + N \nabla\phi = D\bu_S
\,.
\label{def-m-mom}
\end{align}
As we shall see in the Hamiltonian formulation, mapping from the total momentum $\wt{\bM}$ to the in slice fluid momentum $\bm_S$ is known as the \emph{tangling map}. 
The motion equation for the fluid momentum $\bm_S$ in \eqref{def-m-mom} is given by
\begin{align}
\begin{split}
(\partial_t  + \mathcal{L}_{u_s}) \left(\bm_S\cdot d\bx \otimes d^3x\right) 
&= 
\left( 
D d B - D\gamma(z)d \vartheta_s -\frac{s}{2}d\pi_T^2 
- {\rm div} \left(\frac{\delta H_W}{\delta  \bk}\right) d\phi
+ N d \left( \frac{\delta H_W }{\delta  N}\right)
\right)\otimes d^3x
\,,
\end{split}
\label{SVP3-det-redux}
\end{align}
where we have used evolution equation of $\pi_T\nabla \vartheta_s$ and $N\nabla \phi$ through equations in \eqref{eqns: Clebsch1}. Noting that $Dd^2x$ is still advected, from \eqref{SVP3-det-redux} one obtains the equation for the circulation 1-form $D^{-1}\bm_S$ needed for the Kelvin theorem, 
\begin{align}
(\partial_t  + \mathcal{L}_{u_s}) \left( D^{-1}\bm_S\cdot d\bx \right) = 
 d B - \gamma(z)d \vartheta_s -\frac{s}{2D}d\pi_T^2
- \frac{1}{D}{\rm div} \left(\frac{\delta H_W}{\delta  \bk} \right) d\phi
+ \frac{N}{D} d \left( \frac{\delta H_W }{\delta  N}\right)
\,.\label{Kel-form}
\end{align}
The final two terms in equation \eqref{Kel-form} can be evaluated in similar form to the 3D WKB closure of the GLM EB equations in \cite{GH1996,Holm2019}, where we further specialise to the 2D slice domain.
Namely,
\begin{align}
H_W = - \int_M N \omega(\bk) \,d^2x 
\,,\quad\hbox{so that}\quad
 \frac{\delta H_W }{\delta  N}\Big|_{\bk} = - \,\omega(\bk)
 \,,\quad\hbox{and}\quad
 \frac{\delta H_W}{\delta  \bk}\Big|_{N} = - N \frac{\partial \omega(\bk) }{\partial \bk} 
\,,\label{separatedWaveHam-redux}
\end{align}
In Subsection \ref{sec:GLM_vertical_slice}, we will use more concrete wave Hamiltonian $H_W$ derived by the using the WKB closure of GLM.

\begin{remark}[Kelvin circulation theorem for WKB wave-current interaction]
The motion equation \eqref{SVP3-det-redux} implies the following Kelvin circulation theorem for WKB wave-current interaction in the EB vertical slice model dynamics with internal gravity waves,
\begin{align}
\begin{split}
\frac{d}{dt} 
&\oint_{c(u_s)} \frac{1}{D}\bm_S\cdot d\bx
= \oint_{c(u_s)} (\partial_t  + \mathcal{L}_{u_s})\left(\frac{1}{D}\bm_S\cdot d\bx\right)
\\&\qquad = \oint_{c(u_s)} d\left(B-\frac{s\pi_T^2}{2D}\right) - \gamma(z)d \vartheta_s 
+
\underbrace{\
\frac{1}{D}  \bigg(\bk \,{\rm div} \left( N \frac{\partial \omega(\bk) }{\partial \bk}\right) 
+ 
N \nabla \omega(\bk)\bigg) \
}_{\hbox{IGW Forcing}}\hspace{-1mm}
\cdot \,d\bx
\,,
\end{split}
\label{Det-GLM-Kelvin}
\end{align}
where $c(u_s)$ is a material loop moving with the flow velocity $\bu_S(x,z,t)$ in the vertical slice. We remark that the notation $\nabla \omega(\bk)$ denotes the spatial gradient of $\omega$ to avoid confusion. 
In \eqref{Det-GLM-Kelvin}, the fluid quantities $B$ and $\bm$  are defined in equations \eqref{def: gamma-Bwave} and \eqref{def-m-mom}. The wave quantity $\frac{\partial \omega(\bk) }{\partial \bk}$ is defined in equation \eqref{separatedWaveHam-redux}. Besides the usual effect of fluid circulation generated by horizonal components of temperature gradients, ones sees the generation of fluid circulation caused by internal gravity wave quantities which can be evaluated using expressions in equations \eqref{separatedWaveHam-redux}.
\end{remark}

\paragraph{Hamiltonian formulation of the EB VSM with wave dynamics.}

The Hamiltonian formulation of the EB VSM with wave dynamics tracks closely to the EB VSM without waves since the wave degree of freedom have similar Hamiltonian structure as transverse fluid momentum and temperature. As the wave Hamiltonian $H_W$ are separate from the Lagrangian for the EB VSM Lagrangian, one can define the tangled Hamiltonian $\wt{h}_T(\bm_S,{D},\pi_T, \vartheta_s, N, \phi)$ by summing the $H_W$ with the tangled Hamiltian for the EB VSM in \eqref{eqn: tangled EB VSM ham} to have
\begin{align}
\wt{h}_T(m_S,m_T,{D},\vartheta_s, N, \phi; p) := \int_{\cal S} 
\frac{|\bm_S|^2}{2{D}} + \frac{(\pi_T - {D}fx)^2}{2{D}} - {D}\gamma(z)\vartheta_s + p(D-1) - H_W(N,\phi)
\,d^2x_s \,.
\label{eqn: HamVS_1-wave}
\end{align}
\begin{remark}
    We remark that the minus sign of the wave Hamiltonian $H_W$ appearing in \eqref{eqn: HamVS_1-wave} are due to the overall minus sign chosen for the phase space Lagrangian $-N\left(\p_t + \bu_S\cdot \nabla \phi\right) + H_W$ in \eqref{eqn: CLagVar1wave}.
\end{remark}
The variations of this Hamiltonian are given by
\begin{align}
\begin{split}
\delta \wt{h}_T &:= \int_{\cal S} 
\bu_S\cdot \delta \bm_S -  B \,\delta {D} + \left(\bu_S\cdot \nabla \vartheta_s + su_T\right)  \delta \pi_T 
-  \left( {\rm div}({D} \bu_S)  + {D}\gamma(z)  \right) \delta\vartheta_s
\\&\hspace{15mm}
+ \left(\bu_S\cdot \nabla \phi + \frac{\delta H_W}{\delta N}\right)\delta N 
-  \left( {\rm div}(N \bu_S) - \frac{\delta H_W}{\delta \phi}\right)\delta \phi
+ \left(D-1\right)\delta p\,d^2x_s
\,,\end{split}
\label{eqn: Ham1var}
\end{align}
where the expression $B$ is given by, cf. equations \eqref{eqn: CLagVar1} and \eqref{def: gamma-B},
\begin{align}
\begin{split}
\frac{\delta \wt{H}_1}{\delta D } 
&= p - \frac{1}{2{D}^2}|\wt{\bM} + \pi_T\nabla \vartheta_s + N\nabla\phi|^2  
- \frac{s^2}{2{D}^2}\pi_T^2 + \tfrac12(fx)^2 - \gamma(z) \vartheta_s
\\&= p - \tfrac12 |\bu_S|^2 -  \tfrac12 u_T^2 - u_Tfx - \gamma(z)\vartheta_s
= - \frac{\delta \ell}{\delta D } = - B
\,.\end{split}
\label{eqn: Ham1B}
\end{align}
The corresponding Hamiltonian equations are obtained from the following matrix operator which comprises a $s$ weighted symplectic $2\times2$ block in $(\pi_T,\vartheta_s)$ in the centre and symplectic $2\times2$ block in $(N,\phi)$ on the lower left which are coupled to the Lie-Poisson $2\times2$ block in $(\bm_S,{D})$ on the upper left through the semi-direct product structure. Namely, in coordinate form, we have
\begin{equation}
\frac{\p}{\p  t}
\begin{bmatrix}\,
{m_S}_j\\ {D} \\ \pi_T \\ \vartheta_s \\ N \\ \phi
\end{bmatrix}
= -
\begin{bmatrix}
{m_S}_k\partial_{j} + \partial_{k} {m_S}_j &{D}\partial_{j} & \pi_T \p_{j} & -{\vartheta_s}_{,j} & N\p_j & \phi_{, j}
\\
\partial_k {D} & 0 & 0 & 0 & 0 & 0
\\
\p_k \pi_T & 0 & 0 &  -s & 0 & 0
\\
{\vartheta_s}_{, k} & 0 & s & 0 & 0 & 0
\\
\p_k N & 0 & 0 &  0 & 0 & -1
\\
\phi_{, k} & 0 & 0 & 0 & 1 & 0   
\end{bmatrix}
\begin{bmatrix}
{\delta \wt{h}_T/\delta {m_S}_k} = {{u}_s}^k\\
{\delta \wt{h}_T/\delta {D}} =   - B \\
{\delta \wt{h}_T/\delta {\pi_T}} =  su_T  \\
{\delta \wt{h}_T/\delta \vartheta_s} =  - {D}\gamma(z)  \\
{\delta \wt{h}_T/\delta N} = - {\delta H_W}/{\delta N}  \\
{\delta \wt{h}_T/\delta \phi} =  - {\delta H_W}/{\delta \phi}  \\
\end{bmatrix}
\,.
  \label{VSM-LPbrkt-1-tangled}
\end{equation}
To untangle the canonical symplectic structure in the pairs of variables $(\pi_T, \vartheta_s)$ and $(N, \phi)$ from the semi-direct product Lie-Poisson structure of $(\bm_S, {D})$, we consider the untangled Hamiltonian $\wt{h}_{UT}$ in terms of the total momentum $\wt{\bM}$
\begin{align}
\begin{split}
&\wt{h}_{UT}(\wt{\bM},{D},\pi_T, \vartheta_s, N, \phi) 
\\&:= \int_{\cal S} 
\frac{1}{2{D}}|\wt{\bM} + \pi_T\nabla \vartheta_s + N\nabla\phi|^2 
+ \frac{1}{2D}\big(s\pi_T - Dfx\big)^2  - {D}\gamma(z) \vartheta_s + p({D}-1) - H_W(N,\phi)
\,d^2x_s
\,.\end{split}
\label{eqn: HamVS_2-wave}
\end{align}
Then, the Hamiltonian equations corresponding to \eqref{eqn: tangled EB VSM ham} can be written in terms of the following block-diagonal matrix operator form
\begin{equation}
\frac{\p}{\p  t}
\begin{bmatrix}\,
{\wt{M}}_j\\ {D} \\ \pi_T \\ \vartheta_s \\ N \\ \phi
\end{bmatrix}
= - 
   \begin{bmatrix}
   {\wt{M}}_k\partial_{j} + \partial_{k} {\wt{M}}_j &{D}\partial_{j} & 0 & 0 & 0 & 0
   \\
   \partial_k {D} & 0 & 0 & 0 & 0 & 0
      \\
   0 & 0 & 0 &  -1 & 0 & 0
   \\
   0 & 0 & 1 & 0 & 0 & 0
      \\
   0 & 0 & 0 &  0 & 0 & -1
   \\
   0 & 0 & 0 & 0 & 1 & 0   \end{bmatrix}
   \begin{bmatrix}
{\delta \wt{h}_{UT}/\delta {\wt{M}}_k} = {{u}_s}^k\\
{\delta \wt{h}_{UT}/\delta {D}} =   - B \\
{\delta \wt{h}_{UT}/\delta {\pi_T}} = \bu_S\cdot \nabla \vartheta_s + su_T  \\
{\delta \wt{h}_{UT}/\delta \vartheta_s} = - \,{\rm div}({D} \bu_S)  - {D}\gamma(z)  \\
{\delta \wt{h}_{UT}/\delta N} = \bu_S\cdot \nabla \phi - {\delta H_W}/{\delta N}  \\
{\delta \wt{h}_{UT}/\delta \phi} = - \,{\rm div}(N \bu_S)  - {\delta H_W}/{\delta \phi}  \\
\end{bmatrix}
\,,
  \label{VSM-LPbrkt-1-untangled}
\end{equation}
where we note that the Lie-Poisson $2\times2$ block in $(\wt{\bM},{D})$ on the upper left, a symplectic $2\times2$ block in $(\pi_T,\vartheta_s)$ in the centre, and another symplectic $2\times2$ block in $(N,\phi)$ on the lower left are completely separate. 

\subsection{The geometry of the generalised Lagrangian mean (GLM) approach to wave-current interaction}\label{sec: Geom GLM}

One particular wave Hamiltonian of interest is that corresponding to a closure of GLM. Following the initial contribution of Andrews and McIntyre \cite{AM1978}, a number of papers have emerged which seek to describe the generalised Lagrangian mean (GLM) approach to wave mean flow interaction (WMFI) from a geometric perspective. Notable contributions to this field were made in the context of an Euler-Boussinesq fluid \cite{GH1996}, for general Euler-Poincar\'e systems on semidirect product Lie algebras \cite{H2002b}, and, more recently, on any real Riemannian manifold \cite{GV2018}.

\paragraph{A review of WMFI in an Euler-Boussinesq fluid.}

Gjaja and Holm \cite{GH1996} closed the GLM approach to WMFI \cite{AM1978} for stratified rotating incompressible fluid flow in the Euler-Boussinesq (EB) approximation. This was accomplished by assuming a complex vector Wentzel–Kramers–Brillouin (WKB) representation of the internal wave amplitude and performing an asymptotic expansion of Hamilton's principle for the EB model to several orders in the ratio of wave-amplitude to wave-length $\alpha$ and the phase fluctuation ratio $\epsilon$. Recently \cite{HHS2023c}, the GH WKB closure of GLM at quadratic order $O(\alpha^2)$ in the wave amplitude, $\alpha$, has been expressed as a composition of two maps and derived through the semidirect product \cite{HMR1998} variational structure of an Euler-Boussinesq fluid. Understanding this approach by viewing the flow map as a composition of maps gives this model an intrinsic connection to the previous work of the authors on wave current interaction from the viewpoint of a composition of two maps \cite{HHS2023a,HHS2023b,HHS2023c}.

In the geometric description of the Landau paradigm which separates a turbulent flow into fluctuations around a mean flow, the flow map considered here is expressed as the composition of fluctuation map acting via pullback on a mean flow map. Taking the time derivative of the Lagrangian trajectory corresponding to this map results in two velocities, the \emph{Lagrangian mean} and \emph{Lagrangian disturbance} velocities, as introduced in \cite{AM1978}. To achieve the aforementioned closure \cite{GH1996,HHS2023c}, a WKB approximation is made in the displacement part of the Lagrangian trajectory. Performing the resulting approximations and expansions in Hamilton's Principle gives a closed system of equations for the mean flow and wave motion, featuring both the wave effects on currents and the effect of the currents on the waves. In particular, a dispersion relation was found for the Doppler-shifted frequency which extends that of the classical setting. In the classical setting it is assumed that vertical derivatives of the pressure dominate and a \emph{buoyancy frequency} is introduced to represent this effect. However, in the dispersion relation found for wave mean flow interactions \cite{GH1996,HHS2023c}, the effect of the fluid pressure on the (Doppler-shifted) wave frequency is featured in its entirety. As we will see in Proposition \ref{prop:slice_WCI_disp_rel}, a similar such dispersion relation can be found in the vertical slice framework, and its comparison to the standard theory of two dimensional internal gravity waves turns out to be similar.

\subsection{A wave mean flow decomposition of dynamics in the vertical slice}\label{sec:GLM_vertical_slice}

As in the three dimensional case, we begin by considering a flow map which is a composition between a mean flow and a fluctuating part. In particular, we represent spatial path of the flow map as
\begin{equation}
    \begin{pmatrix} x^{\xi}_t \\ y^{\xi}_t \\ z^{\xi}_t \end{pmatrix}= g_t\begin{pmatrix} X \\ Y \\ Z \end{pmatrix} = (Id + \alpha\xi_t)\circ \bar{g}_t\begin{pmatrix} X \\ Y \\ Z \end{pmatrix} = \begin{pmatrix} x_t \\ y_t \\ z_t \end{pmatrix} + \alpha\bxi_t \begin{pmatrix} x_t \\ y_t \\ z_t \end{pmatrix} \,,
    \label{eqn:FlowMap}
\end{equation}
where $(x^{\xi},y^{\xi},z^{\xi})$ are the coordinates of the full trajectory and $(x,y,z)$ are the \emph{mean} coordinates. As in the standard case, we will reserve the bold notation to mean the in-slice coordinates $\bx^{\xi} = (x^{\xi},z^{\xi})$, $\bx = (x,z)$, and $\bX = (X,Z)$.

Here, the mean flow map, $\bar{g}_t$, is assumed to be of the form \eqref{eqn:vertical_slice_flow_map} and corresponds to the \emph{Lagrangian mean velocity}
\begin{equation}
    \dot{\bar{g}}\bar{g}^{-1}\bx_t = \ob{\bu}_t(x_t,z_t) = (\ob{\bu}_S,\ob{u}_T) \,,
\end{equation}
in which the bar denotes Lagrangian mean. The fluctuation term $\bxi$ in \eqref{eqn:FlowMap} is defined to be an array of three scalar displacements in position defined on the vertical slice domain as functions of $x$ and $z$, i.e. $\mcal{F}(M)$. The array of displacements $\bxi$ can be split into its slice and transverse components as $\bxi = (\xi_1,\xi_2,\xi_3) = (\bxi_S,\xi_T)$.

Due to these assumptions and with the above notation, equation \eqref{eqn:FlowMap} can be rewritten as a pair of equations
\begin{align*}
    \bx^{\xi}_t &= g_t\bX = (Id + \alpha\xi_t)\circ \bar{g}_t\bX = \bx_t + \alpha\bxi_S(\bx_t,t)
    \,,\\
    y^{\xi}_t &= g_tY = (Id + \alpha\xi_t)\circ \bar{g}_tY = y_t + \alpha\xi_T(\bx_t,t)
    \,.
\end{align*}

In this set-up, the full map $g_t$ is also a map of the type described by equation \eqref{eqn:vertical_slice_flow_map} and thus can be described by an element of ${\rm Diff}(M)\ltimes \mcal{F}(M)$. Indeed, since $\bar{g}_t$ is of the form \eqref{eqn:vertical_slice_flow_map} and each component of $\bxi$ is a function of $x$ and $z$ only, we have
\begin{equation*}
    \begin{pmatrix} x^{\xi}_t \\ y^{\xi}_t \\ z^{\xi}_t \end{pmatrix} = (Id + \alpha\xi_t)\begin{pmatrix} x_t(X,Z) \\ y_t(X,Z) + Y \\ z_t(X,Z) \end{pmatrix} = \begin{pmatrix} x_t(X,Z) + \alpha\xi_1(x_t(X,Z),z_t(X,Z))\\ y_t(X,Z) + Y + \alpha\xi_2(x_t(X,Z),z_t(X,Z)) \\ z_t(X,Z) + \alpha\xi_3(x_t(X,Z),z_t(X,Z))  \end{pmatrix} \,,
\end{equation*}
which is of the form \eqref{eqn:vertical_slice_flow_map}.
Taking the time derivative of the full map $g_t$ reveals the structure of its tangent velocity field
\begin{equation}
\begin{aligned}    
    \bu(x^{\xi}_t,y^{\xi}_t,z^{\xi}_t,t) = \frac{d(x^{\xi}_t,y^{\xi}_t,z^{\xi}_t)}{dt} &= \ob{\bu}(x_t,z_t,t) + \alpha\left( \p_t\bxi(x_t,z_t,t) + \ob{\bu}_S(x_t,z_t,t)\cdot\nabla \bxi(x_t,z_t,t) \right)
    \\
    &= \ob{\bu}(x_t,z_t,t) + \alpha(\p_t + \mathcal{L}_{\ob{\bu}_S})\begin{pmatrix} \xi_1 \\ \xi_2 \\ \xi_3  \end{pmatrix}\,.
\end{aligned}
\end{equation}
Notice that $\ob{\bu}$ and $\bxi$ are both three dimensional objects defined on our two dimensional domain, $M$. We thus introduce notation  We then make the assumption that the fluctuating part takes the form
\begin{equation}\label{eqn:WKB_fluctuation}
    \bxi(\bx,t) = \ba(\epsilon\bx,\epsilon t)e^{i\phi(\epsilon\bx,\epsilon t)/\epsilon}
+ \ba^*(\epsilon\bx,\epsilon t)e^{-i\phi(\epsilon\bx,\epsilon t)/\epsilon}\,,
\end{equation}
which is motivated by a WKB approximation. The pressure therefore also decomposes into mean and fluctuating parts as
\begin{equation}\label{eqn:WKB_pressure}
        p(\bx^{\xi}, t) = p_0(\bx^{\xi}, t) + \sum_{j\geq 1}\alpha^j\left(b_j(\epsilon\bx^{\xi}, \epsilon t)e^{ij\phi(\epsilon\bx^{\xi},\epsilon t)/\epsilon} + b^*_j(\epsilon\bx^{\xi}, \epsilon t)e^{-ij\phi(\epsilon\bx^{\xi},\epsilon t)/\epsilon}\right)\,.
\end{equation}
Since we have introduced a wave phase variable, $\phi(\bx,t)$, we can define the wave vector, frequency, and doppler-shifted frequency in the standard way
\begin{equation}\label{eqn-DopplerShift-Freq}
    \bk = \nabla\phi \,,\quad \omega = -\p_t\phi \,,\quad\hbox{and}\quad \wt{\omega} = -(\p_t + \bu_S\cdot\nabla)\phi \,.
\end{equation}

As demonstrated in Appendix \ref{appendix:expansion}, the action corresponding to the Euler-Boussinesq Eady model under these approximations is
\begin{equation}\label{action:EB-Eady-WCI}
\begin{aligned}
    S[\ob{\bu}_S,\ob{u}_T,D,\theta_s,p_0,b,\ba] &= \int_{t_0}^{t_1} \ell[\ob{\bu}_S,\ob{u}_T,D,\theta_s,p_0,b,\ba] \,dt 
    \\
    &= \int_{t_0}^{t_1}\int_M \frac{D}{2} \bigg(\big|\ob{\bu}_S\big|^2 + \ob{u}_T^2 + 2\alpha^2 \wt{\omega}^2\left( \ba_S\cdot\ba_S^* + a_Ta_T^*\right) \bigg) 
    \\
    &\qquad\qquad + Df\ob{u}_Tx + Df\alpha^2i\wt{\omega}\left( -a_Ta_1^* + a_T^*a_1 \right) + \frac{g}{\theta_0}D\left( z - \frac{H}{2} \right) \theta_s
    \\
    &\qquad\qquad - D\alpha^2i\left( b\bk\cdot\ba_S^* - b^*\bk\cdot\ba_S \right) - D\alpha^2a^*_ia_j\frac{\p^2p_0}{\p x_i\p x_j} 
    \\
    &\qquad\qquad+ (1-{D})p_0 + \mathcal{O}(\alpha^2\epsilon) \, d^2x\,dt \,,
\end{aligned}
\end{equation}
where we sum over $i,j \in \{1,2\}$. By constraining the relationship between $N$, $\phi$, and $\wt{\omega}$, we have
\begin{equation}\label{action:EB-Eady-WCI-constrained}
\begin{aligned}
    S[\ob{\bu}_S,\ob{u}_T,D,\theta_s,p_0,b,\ba] &= \int_{t_0}^{t_1} \ell[\ob{\bu}_S,\ob{u}_T,D,\theta_s,p_0,b,\ba] \,dt 
    \\
    &= \int_{t_0}^{t_1}\int_M \frac{D}{2} \bigg(\big|\ob{\bu}_S\big|^2 + \ob{u}_T^2 + 2\alpha^2 \wt{\omega}^2\left( \ba_S\cdot\ba_S^* + a_Ta_T^*\right) \bigg) 
    \\
    &\qquad\qquad + Df\ob{u}_Tx + Df\alpha^2i\wt{\omega}\left( -a_Ta_1^* + a_T^*a_1 \right) + \frac{g}{\theta_0}D\left( z - \frac{H}{2} \right) \theta_s
    \\
    &\qquad\qquad - D\alpha^2i\left( b\bk\cdot\ba_S^* - b^*\bk\cdot\ba_S \right) - D\alpha^2a^*_ia_j\frac{\p^2p_0}{\p x_i\p x_j} + (1-{D})p_0 \,d^2x
    \\
    &\qquad\qquad+ \alpha^2 \scp{N}{-\frac{\p}{\p\epsilon t}\phi - \ob{\bu}_S\cdot \nabla_{\epsilon\bx} \phi - \wt{\omega}} + \mathcal{O}(\alpha^2\ep) \, dt \,.
\end{aligned}
\end{equation}
Taking variations in \eqref{action:EB-Eady-WCI-constrained}, we have the following 
\begin{equation}\label{eqn:variations_EB-Eady-WCI-constrained}
    \begin{aligned}
         0 = \delta S = \int_{t_0}^{t_1} &\scp{\delta \ob{\bu}_S}{D\ob{\bu}_S - \alpha^2N\bk } + \scp{\delta\ob{u}_T}{D\ob{u}_T + Dfx} + \scp{\delta\theta_s}{\frac{g}{\theta_0}D\left( z - \frac{H}{2}\right)} 
         \\
         & + \scp{\delta D}{\varpi} + \scp{\delta p_0}{1-D} + \scp{\delta N}{-\frac{\p\phi}{\p \epsilon t} - \ob{\bu}_S\cdot\nabla\phi - \wt{\omega}}
         \\
         &  + \scp{\delta\phi}{\frac{\p N}{\p\epsilon t} + {\rm div}_{\epsilon\bx}(\ob{\bu}_SN) + i{\rm div}_{\epsilon\bx}(Db\ba_S^* - Db^*\ba_S)}
         \\
         & + \scp{\delta\wt{\omega}}{2D\alpha^2\wt{\omega}\left( \ba_S\cdot\ba_S^* + a_Ta_T^*\right) + Df\alpha^2 i (-a_Ta_1^* + a_T^*a_1) - \alpha^2N }
         \\
         & + \scp{\delta a_1^*}{D\alpha^2\wt{\omega}^2 a_1 - iDf\alpha^2\wt{\omega}a_T - iD\alpha^2 b k_1 - \alpha^2D a_j\frac{\p^2 p_0}{\p x_1 \p x_j}} + c.c.  
         \\
         & + \scp{\delta a_2^*}{D\alpha^2\wt{\omega}^2 a_2 - iD\alpha^2 b k_2 - \alpha^2 D a_j\frac{\p^2 p_0}{\p x_2\p x_j}} + c.c.
         \\
         & +\scp{\delta a_T^*}{D\alpha^2 \wt{\omega}^2 a_T + iDf\alpha^2\wt{\omega}a_1} + \scp{\delta b^*}{-\bk\cdot\ba_S} + c.c. \,dt\,,
    \end{aligned}
\end{equation}
where we have introduced a reduced notation for the variation in the density, $D$,
\begin{equation}
\begin{aligned}
    \varpi := \frac{\delta\ell}{\delta D} &= \frac{|\ob{\bu}_S|^2 + \ob{u}_T^2}{2} + f\ob{u}_Tx + \frac{g}{\theta_0}\left( z - \frac{H}{2} \right)\theta_s + \alpha^2\wt{\omega}^2\left( \ba_S\cdot\ba_S^* + a_Ta_T^* \right) - p_0
    \\
    &\qquad + \alpha^2fi\wt{\omega}\left( -a_Ta_1^* + a_T^*a_1 \right) - \alpha^2i\left( b\bk\cdot\ba_S^* - b^*\bk\cdot\ba_S \right) - \alpha^2a_i^*a_j\frac{\p^2p_0}{\p x_i\p x_j} \,.
\end{aligned}
\end{equation}

{\color{black}
As explained in Appendix \ref{appendix:expansion}, the advected quantities featuring in the action \eqref{action:EB-Eady-WCI} are defined along the mean part of the Lagrangian trajectory. For example, a scalar advected quantity satisfies $a_t(\bx^{\xi}_t) = a_t((Id + \alpha\xi_t)\circ \bar{g}_t \bx_0) = a_0(\bx_0)$. Thus, by defining $a_t^{\xi}(\bx_t) = a_t((Id + \alpha\xi_t)\bx_t)$, we see that the quantity $a_t^{\xi}$ is advected by the \emph{mean} flow. This idea can be generalised to advected quantities with a different geometric form by examining how their basis transforms (see Appendix \ref{appendix:expansion}), and the advected quantities, $\theta_s$ and $Dd^2x$, found in the action \eqref{action:EB-Eady-WCI} are of this form. The Euler-Poincar\'e equations given in Remark \ref{rmk:general_vertical_slice_EP} may therefore be applied to our wave current interaction action, where the transport velocity is the in-slice Lagrangian mean velocity. The remaining variations will separately give the wave dynamics. 
}

\paragraph{The total momentum equation. }

The \emph{total momentum} of the system is the $1$-form density defined by
\begin{equation}\label{eqn:total_momentum}
    M = \bs{M}\cdot d{\bx} \otimes d^2x := \frac{\delta\ell}{\delta\ob{u}_S} = D\ob{\bu}_S\cdot d{\bx} \otimes d^2x - \alpha^2N\bk\cdot d{\bx} \otimes d^2x \,.
\end{equation}
We assemble the variational derivatives of the Lagrangian (with respect to $\ob{u}_S$, $\ob{u}_T$, $D$, and $\theta_s$) into the Euler-Poincar\'e equation \eqref{eqn:EP_slice_mom}, we have
\begin{equation*}
    (\p_t + \mathcal{L}_{\ob{u}_S})\left[ \frac{1}{D}\left( D\ob{\bu}_S - \alpha^2N\bk \right)\cdot d{\bx} \right] + \frac{1}{D}\left( D\ob{u}_T + Dfx \right)d\ob{u}_T = d\varpi - \frac{1}{D}\left[ \frac{g}{\theta_0}D\bigg( z - \frac{H}{2}\bigg)\right] \,.
\end{equation*}
For the wave action density, $N$, defined by the variation in the doppler-shfited phase, we have
\begin{equation}\label{eqn:wave_action_density}
    \frac{N}{D} = 2\wt{\omega}\left( \ba_S\cdot\ba_S^* + a_Ta_T^*\right) + if(-a_Ta_1^* + a_T^*a_1) \,,
\end{equation}
and we thus have
wea
\begin{equation}\label{eqn:EP_total_mom_geometric_form}
\begin{aligned}
    (\p_t + \mathcal{L}_{\ob{u}_S})\left( \frac{\bs{M}\cdot d{\bx}}{D} \right) &= d\left( \frac{|\ob{\bu}|^2}{2} - p_0 \right) + f\ob{u}_T\wh{x}\cdot d{\bx} + \frac{g}{\theta_0}\theta_s\wh{z}\cdot d{\bx} 
    \\
    &\quad + \alpha^2d\left( \wt{\omega}\frac{N}{D} - \wt{\omega}^2(\ba_S\cdot\ba_S^* + a_Ta_T^*) - a_i^*a_j\frac{\p^2p_0}{\p x_i\p x_j} \right) \,.
\end{aligned}
\end{equation}
In vector calculus notation\footnote{Note that we have used the identity $\mathcal{L}_u(\bs{A}\cdot d{\bx}) = \left( - \bs{u}\times {\rm curl}\bs{A} + \nabla(\bs{u}\cdot\bs{A}) \right)\cdot d{\bx}$}, the equation of motion is
\begin{equation}\label{eqn:EP_total_mom}
\begin{aligned}
    \p_t\frac{\bs{M}}{D} - \ob{\bu}_S\times {\rm curl}\frac{\bs{M}}{D} &+ \nabla\left( \frac{|\ob{\bu}_S|^2}{2} + p_0 \right) - f\ob{u}_T\wh{x} - \frac{g}{\theta_0}\theta_s\wh{z}
    \\
    & +\alpha^2\nabla \left( -\omega \frac{N}{D} + \wt{\omega}^2(\ba_S\cdot\ba_S^* + a_Ta_T^*) + a_i^*a_j\frac{\p^2p_0}{\p x_i\p x_j} \right) = 0
    \,,
\end{aligned}
\end{equation}
where the term $\nabla(\ob{\bu}_S \cdot (\bs{M} / D))$ arising from the Lie derivative has been combined with the remainder of the equation, in particular serving to remove the Doppler shift from the frequency appearing in the first term of order $\alpha^2$. 

\paragraph{The mean flow momentum equation. }

We decompose the total momentum \eqref{eqn:total_momentum} into a mean flow momentum and wave momentum
\begin{equation}
    \bs{M} = D\ob{\bu}_S - \alpha^2N\bk =: \bs{m} - \bp \,.
\end{equation}
From the variational principle we have equations for $N$ and $\phi$,
\begin{align}
    \p_t N + {\rm div}(N \ob{\bu}_S) &= - i{\rm div}(Db\ba_S^* - Db^*\ba_S)
    \,,\\
    \p_t\phi + \ob{\bu}_S\cdot\nabla \phi &= -\wt{\omega}
    \,,
\end{align}
which, when combined, give an equation for the wave momentum
\begin{equation}
    (\p_t + \mathcal{L}_{\ob{\bu}_S})(\bp\cdot d\bx) = \alpha^2(\p_t + \mathcal{L}_{\ob{\bu}_S})(Nd\phi) = -\alpha^2Nd\wt{\omega} -\alpha^2{\rm div}\left( N\bs{v}_g\right)d\phi \,,
\end{equation}
where
\begin{equation*}
    \bs{v}_g := \frac{iD}{N}\left( \ba_S^*b - \ba_Sb^* \right) \,.
\end{equation*}
Combining this equation with the Euler-Poincar\'e equation \eqref{eqn:EP_total_mom_geometric_form} gives the equation for the mean flow momentum
\begin{equation}\label{eqn:EP_mean_mom}
\begin{aligned}
    \p_t\ob{\bu}_S - \ob{\bu}_S\times{\rm curl}\ob{\bu}_S &= - \nabla\left( \frac{|\ob{\bu}_S|^2}{2} + p_0 \right) + f\ob{u}_T\wh{x} + \frac{g}{\theta_0}\theta_s\wh{z}
    \\
    &\quad - \alpha^2\nabla\left( -\wt{\omega} \frac{N}{D} + \wt{\omega}^2(\ba_S\cdot\ba_S^* + a_Ta_T^*) + a_i^*a_j\frac{\p^2p_0}{\p x_i\p x_j} \right)
    \\
    &\quad -\frac{\alpha^2}{D}\left(N\nabla\wt{\omega} + \bk {\rm div}\left( N\bs{v}_g\right) \right)
    \,.
\end{aligned}
\end{equation}
Notice that by moving the time derivative of the wave momentum, $(\p_t + \mathcal{L}_{\ob{\bu}_S})(\bp\cdot d\bx)$, to the right hand side of the equation of motion, we no longer have the term in our Lie derivative which removed the Doppler shift from the frequency in equation \eqref{eqn:EP_total_mom}. Thus, all occurrences of the wave frequency in equation \eqref{eqn:EP_mean_mom} are Doppler shifted.

\paragraph{The equation for the mean transverse velocity. }

The Euler-Poincar\'e equation \eqref{eqn:EP_slice_transverse} from Remark \ref{rmk:general_vertical_slice_EP}, when combined with Hamilton's Principle \eqref{eqn:variations_EB-Eady-WCI-constrained}, immediately gives the equation for the mean transverse velocity
\begin{equation}
    \p_t\ob{u}_T + \ob{\bu}_S\cdot\nabla \ob{u}_T + f\ob{\bu}_S\cdot\wh{x} = - \frac{g}{\theta_0}\bigg(z - \frac{H}{2} \bigg)s \,.
\end{equation}
\begin{remark}
    This equation is identical to that present in the standard Euler-Boussinesq Eady model \cite{CH2013}. This is to be expected, since each element of the fluctuating term $\bxi$ in \eqref{eqn:FlowMap} is a function on $M$ and thus has no derivative in the direction transverse to the slice. As such, the wave effect on current occurs only on the velocity field $\ob{\bu}_S$ within the vertical slice.
\end{remark}

\paragraph{The advection equations and incompressibility. }

The variation in $p_0$ implies that $D=1$ up to order $\alpha^2\epsilon$. If we combine this with the advection equation for the mass density, we have incompressibility
\begin{equation}
     \left.\begin{aligned} D &= 1 \\ \p_t D + {\rm div}(D\ob{\bu}_S) &= 0 \end{aligned}\right\} \quad\implies\quad \nabla\cdot \ob{\bu}_S = 0 \,.
\end{equation}
These equations are to be considered together in tandem with the advection equation for the scalar buoyancy,
\begin{equation}
    \p_t\theta_s + \ob{\bu}_S\cdot\nabla\theta_s + \ob{u}_Ts = 0 \,.
\end{equation}

\paragraph{The wave dynamics. }

Using the variations of $\delta a_1$, $\delta a_2$, $\delta a_T$, $\delta b$ and the stationary condition, the linear set of equations 
\begin{align}
    \begin{split}
        \wt{\omega}^2 a_1 - if\wt{\omega}a_T - ib k_1 - a_j\frac{\p^2 p_0}{\p x_1 \p x_j} &= 0 \,,\\
        \wt{\omega}^2 a_2 - ib k_2 - a_j\frac{\p^2 p_0}{\p x_2 \p x_j} &= 0\,,\\
        \wt{\omega}^2 a_T + if\wt{\omega}a_1 &= 0\,,\\
        \bk\cdot\ba_S &= 0\,,
    \end{split} \label{eq:ab constrains}
\end{align}
implies a linear dispersion relation for the Doppler shifted wave frequency $\wt{\omega}$ 
defined in \eqref{eqn-DopplerShift-Freq} 
\begin{equation}\label{eqn:slice_WCI_dispersion_rel}
    \wt{\omega}^2 = \frac{f^2k_2^2}{|\bk|^2} + \left( \delta_{ij} - \frac{k_ik_j}{|\bk|^2} \right)\frac{\p^2p_0}{\p x_i\p x_j} \,.
\end{equation}
{\color{black}
The details of the derivation of the dispersion relation equation in \eqref{eqn:slice_WCI_dispersion_rel} are shown in the proof of Proposition \ref{prop:slice_WCI_disp_rel} in Appendix \ref{app-B-Dispersion}.
}
\begin{remark}
The dispersion relation \eqref{eqn:slice_WCI_dispersion_rel} takes a more familiar form when reverting many of our modelling assumptions back to a more standard form. In particular, we have assumed that the waves are coupled to the mean flow and the frequency is Doppler shifted. We have assumed further that the pressure here is the complete pressure required to ensure that the flow is incompressible. If we instead consider the non-shifted frequency and make the standard assumption that the pressure increases with depth such that $p_0$ only has derivatives in $z$ and is related to the buoyancy frequency, ${\cal N}$, in the standard way, the dispersion relation in \eqref{eqn:slice_WCI_dispersion_rel} simplifies to the following form,
\begin{equation}\label{eqn-IW-dispersion}
    \left.\begin{aligned} \wt{\omega} &\mapsto \omega \\ p_0 &= \frac{{\cal N}^2z^2}{2}\end{aligned}\right\} \quad\implies\quad \omega^2 = \frac{k_2^2f^2 + k_1^2{\cal N}^2}{|\bk|^2}\,.
\end{equation}
Under these assumptions the dispersion properties of the waves found in our model are a generalisation of the classical theory for internal gravity waves. This is clear, since the approximate dispersion relation in \eqref{eqn-IW-dispersion} may be found in standard texts on fluid dynamics \cite{V2017}. This generalisation is noteworthy since the assumption that vertical gradients of the pressure are dominant is not universally valid.
\end{remark}




\section{Stochastic advection by Lie transport (SALT) for VSMs}\label{sec:SALTyVSM}
As discussed at the beginning of Section \ref{sec:WMFI}, parameterisations of fast fluctuations are essential in the numerical simulation of geophysical fluid dynamics. A more modern approach is through the application of stochastic parametersiation schemes, where one obtains a statistical representation of uncertainty during ensemble forecasting simulations. In this section, we briefly consider a stochastic parameterisation scheme known as SALT \cite{H2015} applied to the family of VSMs. 

Recall that the Euler-Poincar\'e equations resulting from Theorem \ref{thm:EP_slice} can be written as the following Lie-Poisson equations on $(\mathfrak{X}(M) \ltimes \mathcal{F}(M))^* \times V_1^* \times V_2^*$,
\begin{equation}\tag{\eqref{Eqn: Tangled-Deterministic} revisited}
\frac{\p}{\p t}
\begin{bmatrix}\,
m_S \\ {D} \\ m_T \\ \vartheta_s
\end{bmatrix}
= - 
   \begin{bmatrix}
   \ad^*_{\Box}m_S & \Box \diamond {D} & \Box \diamond m_T  & \Box \diamond \vartheta_s
   \\
   \mathcal{L}_{\Box}{D} & 0 & 0 & 0 
   \\
   \mathcal{L}_{\Box}m_T & 0 & 0 & -s 
   \\
   \mathcal{L}_{\Box}\vartheta_s & 0 & s & 0
   \end{bmatrix}
   \begin{bmatrix}
	{\delta h}/{\delta m_S} \\
	{\delta h}/{\delta {D}} \\
	{\delta h}/{\delta m_T} \\
	{\delta h}/{\delta \vartheta_s}
\end{bmatrix} 
\,,
\end{equation}
where the reduced Hamiltonian, $h(m_S,{D},m_T,\theta_s)$, is related to the Lagrangian through the Legendre transform
\begin{equation}
    h(m_S,{D},m_T,\vartheta_s) = \scp{m_S}{u_S} + \scp{m_T}{u_T} - \ell(u_S,u_T,{D},\vartheta_s,p) \,.
\end{equation}
\begin{remark}
    As for the velocity field, we will denote the momentum by $m_S$, and the coefficients of the momentum with respect to the geometric basis by $\bm_S$.
\end{remark}
As noted in \cite{Street2021}, particular care needs to be taken in the variational description of incompressible fluids with stochastic transport. That is, the variation in $D$ gains a stochastic part corresponding to the stochastic part of the pressure. To achieve this, we define stochastic Hamiltonians, $h_i:\mathfrak{X}^* \ltimes (\Lambda^2 \oplus \Lambda^2 \oplus \Lambda^0 )\rightarrow \mathbb{R}$, by
\begin{equation}
    h_i = \scp{m_S}{\xi_{Si}}_{\mathfrak{X}^*\times\mathfrak{X}} + \int_M \xi_{Ti}\,m_T - \int_M p_i(D-1)\,d^2x \,,
\end{equation}
where $\xi_{Si}\in \mathfrak{X}(M)$ and $\xi_{Ti}\in \Omega^0(M)$. The stochastic Lie-Poisson system is then given by
\begin{equation}\label{eqn:stochastic_LP_sigma}
    \rmd f = \{ f , h \}^\Sigma \,dt + \sum_i \{ f , h_i \}^\Sigma \circ dW_t^i \,,
\end{equation}
where the bracket $\{\cdot,\cdot,\}^\Sigma$ is the bracket also used in the deterministic equation \eqref{eqn:LP_demonstration_conclusion}. We will employ the following shorthand notation
\begin{equation}
    \rmd P := p\,dt + \sum_i p_i\circ dW_t^i \,,\quad \rmd x_S := u_S \,dt + \sum_i \xi_{Si}\circ dW_t^i \,,\quad \rmd x_T := u_T\,dt + \sum_i \xi_{Ti} \circ dW_t^i \,.
\end{equation}
The stochastic Lie-Poisson equations \eqref{eqn:stochastic_LP_sigma} for vertical slice dynamics, can be written in the following explicit geometric form as, cf. equation \eqref{Eqn: Tangled-Deterministic},
\begin{equation}\label{eqn:Stochastic_Lie_Poisson}
\rmd
\begin{bmatrix}\,
m_S \\ {D} \\ m_T \\ \vartheta_s
\end{bmatrix}
= - 
   \begin{bmatrix}
   \ad^*_{\Box}m_S & \Box \diamond {D} & \Box \diamond m_T  & \Box \diamond \vartheta_s
   \\
   \mathcal{L}_{\Box}{D} & 0 & 0 & 0 
   \\
   \mathcal{L}_{\Box}m_T & 0 & 0 & -s 
   \\
   \mathcal{L}_{\Box}\vartheta_s & 0 & s & 0
   \end{bmatrix}
   \begin{bmatrix}
	\rmd x_S := u_S \,dt + \sum_i \xi_{Si}\circ dW_t^i   \\
	{\delta h}/{\delta {D}}\,dt + \sum_i{\delta h_i}/{\delta {D}}\circ dW_t^i \\
	\rmd x_T := u_T\,dt + \sum_i \xi_{Ti} \circ dW_t^i \\
	{\delta h}/{\delta \vartheta_s}\,dt 
\end{bmatrix} 
.
\end{equation}
Thus, we have
\begin{align}
	(\rmd + \mathcal{L}_{\rmd x_S})m_S &= -\rmd x_T\diamond m_T\,dt - \frac{\delta h}{\delta\vartheta_s}\diamond\vartheta_s\,dt - \frac{\delta h}{\delta {D}}\diamond {D}\,dt - \sum_i\frac{\delta h_i}{\delta {D}}\diamond {D}\circ dW_t^i 
    \label{eqn:slice_LP_mom}
	\,,\\
	(\rmd + \mathcal{L}_{\rmd x_S})m_T &= \frac{\delta H}{\delta\theta_s}s\,dt
	\,,\\
	(\rmd + \mathcal{L}_{\rmd x_S})\vartheta_s &= -s\,\rmd x_T
	\,,\\
	(\rmd + \mathcal{L}_{\rmd x_S})({D}d^2x) &= 0
	\,.
\end{align}

\begin{remark}
    Notice that in this model the Poisson bracket remains unmodified, which is evident from observing equations \eqref{eqn:stochastic_LP_sigma} and \eqref{eqn:Stochastic_Lie_Poisson}.
\end{remark}

\subsection{The Euler-Boussinesq Eady model with SALT}

We first Legendre transform the Lagrangian for the Euler-Boussinesq Eady model to get
\begin{equation}
	h(m_S,m_T,\theta_s,{D}) = \int_M \frac{1}{2{D}}\left( |\bs{m}_S|^2 + m_T^2 \right) - {D}fm_Tx - \frac{g}{\theta_0}{D}\left( z - \frac{H}{2} \right)\theta_s + p(D-1)\,d^2x \,.
\end{equation}
The equation \eqref{eqn:slice_LP_mom} corresponding to the above Hamiltonian is
\begin{equation}
    (\rmd + \mathcal{L}_{\rmd x_S})u_S^{\flat} + (u_T+ fx)d(\rmd x_T) = \frac{1}{{D}}\frac{\delta H}{\delta\vartheta_s}d\vartheta_s\,dt -dp\,dt - \sum_i dp_i \circ dW_t^i \,,
\end{equation}
where we have divided by the volume form and used the fact that $\frac{1}{{D}}\frac{\delta\ell}{\delta u_T} = \frac{m_T}{{D}} = u_T + fx$. Note that in the above equation, the exterior derivative is denoted by $d$ and the time increment by $\rmd$.

In vector calculus notation, we have
\begin{align}
	\begin{split}
	\rmd \bu_S + \bu_S \cdot \nabla \bu_S\,dt + \sum_i \bxi_{Si}\cdot\nabla\bu_S \circ dW_t^i  + \sum_i u_j\nabla\xi_{Si}^j \circ dW_t^i  + \sum_i u_T&\nabla \xi_{Ti}\circ dW_t^i + \sum_i fx\nabla \xi_{Ti}\circ dW_t^i 
	\\
	&= fu_T\wh{x}\,dt+ \frac{g}{\vartheta_0}\vartheta_s\wh{z}\,dt - \nabla\rmd P \,,
	\end{split}
	\\
	\rmd u_T + \bu_S \cdot \nabla u_T \,dt + \sum_i \bxi_{Si}\cdot\nabla u_T \circ dW_t^i &= - f\bu_S\cdot\wh{x}\,dt -\frac{g}{\vartheta_0}\left( z - \frac{H}{2} \right)s\,dt
	\,,\\
	\rmd\vartheta_s + \bu_S \cdot \nabla \vartheta_s \,dt + \sum_i \bxi_{Si}\cdot\nabla \vartheta_s \circ dW_t^i + u_Ts\,dt + \sum_i \xi_{Ti}s\circ dW_t^i &= 0
	\,,\\
	\nabla \cdot \bu_S = \nabla \cdot \bxi_{Si} &= 0
	\,.
\end{align}
If we use the reduced notation $\rmd \bx_S$ and $\rmd \bx_T$, and introduce the notation $\bu = (\bu_S,u_T)$ and $\bxi = (\bxi_S,\xi_T)$, then the equations can be written in the following more compact form
\begin{align}
	\rmd \bu_S + \rmd \bx_S\cdot \nabla \bu_S + \sum_i u_j \nabla \xi^j \circ dW_t^i + \sum_i fx\nabla \xi_{Ti}\circ dW_t^i &= fu_T\wh{x}\,dt+ \frac{g}{\vartheta_0}\vartheta_s\wh{z}\,dt - \nabla\rmd P
	\,,\\
	\rmd u_T + \rmd \bx_S \cdot \nabla u_T &= - f\bu_S\cdot\wh{x}\,dt -\frac{g}{\vartheta_0}\left( z - \frac{H}{2} \right)s\,dt
	\,,\\
	\rmd\vartheta_s + \rmd \bx_S\cdot \nabla \vartheta_s + s\,\rmd \bx_T&= 0
	\,,\\
	\nabla \cdot \bu_S = \nabla \cdot \bxi_{Si} &= 0
	\,.
\end{align}
\begin{theorem}[Kelvin-Noether]
	For a vertical slice model with SALT as introduced in equation \eqref{eqn:Stochastic_Lie_Poisson}, we have that
	\begin{equation}
		\rmd \oint_{\gamma_t} \left( s\frac{m_S}{{D}}  - \frac{m_T}{{D}}\nabla\vartheta_s \right)\cdot d\bx = 0 \,,
	\end{equation}
	where $\gamma_t:C^1 \mapsto M$ is a closed loop moving with the flow $\rmd x_S$.
\end{theorem}
\begin{proof}
    The proof of this theorem follows directly from the Lie-Poisson equations \eqref{eqn:Stochastic_Lie_Poisson}. Indeed,
	\begin{equation*}
	\begin{aligned}
		\rmd \oint_{\gamma_t} \left( s\frac{m_S}{{D}}  - \frac{m_T}{{D}}\nabla\vartheta_s \right)\cdot d\bx &= \oint_{\gamma_t}(\rmd + \mathcal{L}_{\rmd x_S})\left[ \left( s\frac{m_S}{{D}}  - \frac{m_T}{{D}}\nabla\vartheta_s \right) \cdot d\bx\right]
		\\
		&= \oint_{\gamma_t} s\left( -d\left(\frac{\delta H}{\delta {D}}\right) \,dt + s\,d\rmd P + \frac{1}{{D}}\frac{\delta H}{\delta\vartheta_s}d\vartheta_s \,dt - \frac{1}{{D}}m_Td(\rmd x_T) \right)
        \\
        &\qquad\qquad - \frac{s}{{D}}\frac{\delta H}{\delta \vartheta_s} d\vartheta_s\,dt + \left(\frac{1}{{D}}m_t\right)d(s\,\rmd x_T)
		\\
		&= 0
		\,,
	\end{aligned}
	\end{equation*}
	as required.
\end{proof}

\begin{remark}
    The statement of the Kelvin-Noether theorem, when written in terms of the Lagrangian, is
	\begin{equation*}
		\rmd \oint_{\gamma_t} \left( s\bigg( \frac{1}{{D}}\frac{\delta\ell}{\delta u_S} \bigg) - \bigg( \frac{1}{{D}}\frac{\delta\ell}{\delta u_T} \bigg)\nabla\vartheta_s \right)\cdot d\bx = 0 \,.
	\end{equation*}
\end{remark}
\begin{remark}
In this section, we have included SALT noise in the vertical slice model. It is possible to go further, and include SALT noise in the WMFI models discussed in the previous section. Indeed, SALT noise can be coupled to each momentum term in equation \eqref{def-m-mom}, just as we have coupled noise terms to the slice and transverse momenta in this section.

\end{remark}

\section{Conclusion and Outlook}\label{sec: Conclude Outlook}

The ocean modelling goal of the present paper has been to include internal gravity wave (IGW) motion in the Vertical Slice Model (VSM) with transverse flow introduced in \cite{CH2013}. The mathematical goal accompanying the goal for data assimilation is also to determine the Poisson/Hamiltonian structure of the resulting model, and thereby formulate a new model of stochastic parameterisation of advective transport which is extremely versatile and therefore is of potential use for quantifying uncertainty in a variety of ocean model simulations, in addition to the stochastic VSM derived here to illustrate the variational ocean modelling approach.

Of course, many fundamental mathematical questions about the stochastic VSM derived here still remain open. Indeed, even the fundamental properties of existence, uniqueness and well‐posedness have not yet been proven for the stochastic VSM, although the stability, well-posedness and blow-up criterion have been established for the deterministic case in \cite{A-OBdL2019}. Nonetheless, the stochastic VSMs do belong to the SALT class of fluid models and therefore they have the capability to quantify uncertainty. In fact, the SALT models are also known to be effective in systematically reducing the uncertainty in the spread of their stochastic ensembles, when the data science methods of calibration and assimilation in \cite{CCHOS18a,CCHOS20} are applied. The application of the uncertainty reduction capabilities of the SALT models and their data calibration and assimilation results for these stochastic VSMs remain to be demonstrated elsewhere. 

\color{black}

\section*{Acknowledgements}
We are grateful to our friends, colleagues and collaborators for their advice and encouragement in the matters treated in this paper. 
We especially thank C. Cotter, I. Gjaja, J.C. McWilliams, C. Tronci, P. Bergold, and J. Woodfield for many insightful discussions of corresponding results similar to the ones derived here for the VSMs, and in earlier work together in deriving models of complex fluids, turbulence, plasma dynamics, vertical slice models and the quantum--classical hydrodynamic description of molecules. 
DH and RH were partially supported during the present work by Office of Naval Research (ONR) grant award  N00014-22-1-2082, ``Stochastic Parameterization of Ocean Turbulence for Observational Networks''. DH and OS were partially supported during the present work by European Research Council (ERC) Synergy grant ``Stochastic Transport in Upper Ocean Dynamics" (STUOD) -- DLV-856408.

\begin{appendices}

\section{The asymptotic expansion}\label{appendix:expansion}

This appendix contains an asymptotic expansion which reveals the form of the action for the dynamical system studied in Section \ref{sec:GLM_vertical_slice}. The appendix follows a sequence of calculations similar to those for the isotropic three-dimensional case in \cite{HHS2023c}. For full details of the set up, see \cite{GH1996}.

The modelling assumptions made in Section \ref{sec:GLM_vertical_slice} decompose a fluid trajectory into the sum of its Lagrangian mean flow trajectory $\bx_t$ and a Lagrangian mean fluctuation displacement $\alpha\bxi_t(\bx_t)$  of relative amplitude $\alpha$ at each point along the mean trajectory $\alpha$ as
\begin{equation*}
    \begin{pmatrix} x^{\xi}_t \\ y^{\xi}_t \\ z^{\xi}_t \end{pmatrix} = \begin{pmatrix} x_t \\ y_t \\ z_t \end{pmatrix} + \alpha\bxi_t \begin{pmatrix} x_t \\ y_t \\ z_t \end{pmatrix} \,,
\end{equation*}
and the associated velocity decomposes into
\begin{equation*}
    \bu(x^{\xi}_t,y^{\xi}_t,z^{\xi}_t,t) = \ob{\bu}(\bx_t,t) + \alpha\left( \p_t\bxi + \ob{\bu}_S\cdot\nabla \bxi \right) \,.
\end{equation*}
The fluctuating terms have a WKB structure, and as such the pressure has a contribution of this form also, c.f. equations \eqref{eqn:WKB_fluctuation} and \eqref{eqn:WKB_pressure},
\begin{align}
    \bxi(\bx_S,t) &= \ba(\epsilon\bx_S,\epsilon t)e^{i\phi(\epsilon\bx_S,\epsilon t)/\epsilon}
+ \ba^*(\epsilon\bx_S,\epsilon t)e^{-i\phi(\epsilon\bx_S,\epsilon t)/\epsilon}\,,
\\
    p(\bX_S, t) &= p_0(\bX_S, t) + \sum_{j\geq 1}\alpha^j\left(b_j(\epsilon\bX_S, \epsilon t)e^{ij\phi(\epsilon\bX_S,\epsilon t)/\epsilon} + b^*_j(\epsilon\bX_S, \epsilon t)e^{-ij\phi(\epsilon\bX_S,\epsilon t)/\epsilon}\right)\,,
\end{align}
where for brevity we have denoted the coordinates in the slice by $\bX_S = (X,Z)$ and $\bx_S = (x,z)$. Recall that the Lagrangian for the Euler-Boussinesq Eady model is given by
    \begin{equation}\tag{\ref{eqn:Eady_action} revisited}
        \ell[\bu_S,u_T,\mathscr{D},\vartheta_s,p] = \int_M \frac{\mathscr{D}}{2}(|\bu_S|^2 + u_T^2) + \mathscr{D}fu_TX + \frac{g}{\vartheta_0}\mathscr{D}\left( Z - \frac{H}{2} \right) \vartheta_s + p(1-\mathscr{D})\, d^2X \,,
    \end{equation}
and we will make the above approximations \emph{within} Hamilton's principle for the action corresponding to this Lagrangian. Noting that the variables defined in the Lagrangian are defined on the coordinates $(X,Z)$, rather than the full three dimensional space, due to the identification of the subgroup of suitable diffeomorphisms being isomorphic to ${\rm Diff}(M)\ltimes \mcal{F}(M)$.

Firstly, we consider how advected quantities are expressed in terms of $\bx_t$. The volume form becomes
\begin{equation}
    \mathscr{D}(X,Z)d^2X =: \mathscr{D}^{\bxi_S}(x,z)d^2X = \mathscr{D}^{\bxi_S}(x,z)\mathscr{J}d^2x =: {D} d^2x \,,
\end{equation}
where we have introduced the helpful notation $\mathscr{D}^{\bxi_S}$ and $$\mathscr{J} = {\rm det}\left( \delta_{ij} + \alpha\frac{\p\xi^i}{\p x^j}\right) \,.$$ Similarly, our scalar advected variable, $\vartheta_s$, transforms as $$ \vartheta_s(X,Z) = \vartheta_s^{\bxi_S}(x,z) =: \theta_s \,.$$

Beginning with the kinetic energy term, we have (up to order $\alpha^2\epsilon$)
\begin{align*}
    \int_M \frac{\mathscr{D}}{2}\left( |\bu_S|^2 + u_T^2 \right) \,d^2X &= \int_M \frac{D}{2}\bigg( \Big|\ob{\bu}_S + \alpha( \p_t\bxi_S + \ob{\bu}_S\cdot\nabla \bxi_S)\Big|^2  + \left(\ob{u}_T + \alpha ( \p_t\xi_T + \ob{\bu}_S\cdot\nabla \xi_T \right)^2 \bigg) \,d^2x
    \\
    &\hspace{-70pt}= \int_M \frac{D}{2}\bigg( \big|\ob{\bu}_S\big|^2 + \alpha^2\big| \ba_Sie^{\frac{i\phi}{\epsilon}}\p_t\phi - \ba_S^*ie^{-\frac{i\phi}{\epsilon}}\p_t\phi + \ba_S ie^{\frac{i\phi}{\epsilon}}\ob{\bu}_S\cdot\nabla\phi - \ba_S^*ie^{-\frac{i\phi}{\epsilon}}\ob{\bu}_S\cdot\nabla\phi \big|^2 
    \\
    &\hspace{-70pt} \qquad\qquad + \ob{u}_T^2 + \alpha^2\left( a_Tie^{\frac{i\phi}{\epsilon}}\p_t\phi - a_T^*ie^{-\frac{i\phi}{\epsilon}}\p_t\phi + a_T ie^{\frac{i\phi}{\epsilon}}\ob{\bu}_S\cdot\nabla\phi - a_T^*ie^{-\frac{i\phi}{\epsilon}}\ob{\bu}_S\cdot\nabla\phi \right)^2 \bigg)\,d^2x
    \\
    &\hspace{-70pt}= \int_M \frac{D}{2}\bigg( \big|\ob{\bu}_S\big|^2 + \ob{u}_T^2 + \alpha^2\big| -\ba_Si\wt{\omega}e^{\frac{i\phi}{\epsilon}} + \ba_S^*i\wt{\omega}e^{-\frac{i\phi}{\epsilon}}\big|^2 + \alpha^2\left(  -a_Ti\wt{\omega}e^{\frac{i\phi}{\epsilon}} + a_T^*i\wt{\omega}e^{-\frac{i\phi}{\epsilon}}\right)^2 \bigg)\,d^2x
    \\
    &\hspace{-70pt}= \int_M \frac{D}{2} \bigg(\big|\ob{\bu}_S\big|^2 + \ob{u}_T^2 + 2\alpha^2 \wt{\omega}^2\left( \ba_S\cdot\ba_S^* + a_Ta_T^*\right) \bigg) \,d^2x \,.
\end{align*}
Denoting $\bxi = (\xi_1,\xi_2,\xi_3)$, the rotation term is
\begin{align*}
    \int_M \mathscr{D}fu_T X \,d^2X &= \int_M Df\left( \ob{u}_T + \alpha\left( -a_Ti\wt{\omega}e^{\frac{i\phi}{\epsilon}} + a_T^*i\wt{\omega}e^{-\frac{i\phi}{\epsilon}} \right)(x + \alpha \xi_1) \right) \,d^2x
    \\
    &= \int_M Df\ob{u}_Tx + Df\alpha^2i\wt{\omega}\left( -a_Ta_1^* + a_T^*a_1 \right) \,d^2x \,.
\end{align*}
The potential energy is handled similarly
\begin{align*}
    \int_M \frac{g}{\vartheta_0}\mathscr{D}\left( Z - \frac{H}{2} \right) \vartheta_s \,d^2X &= \int_M \frac{g}{\theta_0}D\left( z + \alpha\xi_3 - \frac{H}{2} \right) \theta_s \,d^2x = \int_M \frac{g}{\theta_0}D\left( z - \frac{H}{2} \right) \theta_s \,d^2x \,.
\end{align*}
As in the three dimensional case, the pressure terms are the most complex part of the calculation. Since $p$ is a function of the vertical slice coordinates, the expansion of $p(X,Z) = p^{\bxi_S}(x,z)$ only invokes the fluctuations which lie within $M$, namely $\bxi_S$, and derivatives of these with respect to $x$ and $z$. The expansion for the pressure term is therefore equivalent to that found in \cite{HHS2023c}, and the action is
\begin{equation}
\begin{aligned}
    S[\ob{\bu}_S,\ob{u}_T,D,\theta_s,p_0,b,\ba] &= \int_{t_0}^{t_1} \ell[\ob{\bu}_S,\ob{u}_T,D,\theta_s,p_0,b,\ba] \,dt 
    \\
    &= \int_{t_0}^{t_1}\int_M \frac{D}{2} \bigg(\big|\ob{\bu}_S\big|^2 + \ob{u}_T^2 + 2\alpha^2 \wt{\omega}^2\left( \ba_S\cdot\ba_S^* + a_Ta_T^*\right) \bigg) 
    \\
    &\qquad\qquad + Df\ob{u}_Tx + Df\alpha^2i\wt{\omega}\left( -a_Ta_1^* + a_T^*a_1 \right) + \frac{g}{\theta_0}D\left( z - \frac{H}{2} \right) \theta_s
    \\
    &\qquad\qquad - D\alpha^2i\left( b\bk\cdot\ba_S^* - b^*\bk\cdot\ba_S \right) - D\alpha^2a^*_ia_j\frac{\p^2p_0}{\p x_i\p x_j} 
    \\
    &\qquad\qquad+ (1-{D})p_0 + \mathcal{O}(\alpha^2\epsilon) \,, d^2x\,dt \,,
\end{aligned}
\end{equation}
where $\bk = \nabla\phi = (\p_{\epsilon x}\phi,\p_{\epsilon z}\phi)$, and we sum over $i,j \in \{1,2\}$.

\section{Deriving the dispersion relation}\label{app-B-Dispersion}

Recall the linear equations
\begin{align}
    \begin{split}
        \wt{\omega}^2 a_1 - if\wt{\omega}a_T - ib k_1 - a_j\frac{\p^2 p_0}{\p x_1 \p x_j} &= 0 \,,\\
        \wt{\omega}^2 a_2 - ib k_2 - a_j\frac{\p^2 p_0}{\p x_2 \p x_j} &= 0\,,\\
        \wt{\omega}^2 a_T + if\wt{\omega}a_1 &= 0\,,\\
        \bk\cdot\ba_S &= 0\,.
    \end{split} \tag{\eqref{eq:ab constrains} revisited}
\end{align}
Taking the dot product of the first two equations with $\bk$ and using $\bk\cdot \ba_S = 0$ gives the equation for $b$ in terms of $a_1$
\begin{align*}
    ib|\bk|^2 = -f^2 a_1 k_1 - k_i a_j \frac{\p^2 p_0}{\p x_i \p x_j}\,,
\end{align*}
where we have substituted $a_T$ in terms of $a_1$ using $\wt{\omega }a_T = -ifa_1$. Then, we have a set of linear equations involving $a_1$ and $a_2$
\begin{align}
    \begin{split}
        \wt{\omega}^2 a_1 - f^2 a_1 - a_j\frac{\p^2 p_0}{\p x_1 \p x_j} + \frac{k_1}{|\bk|^2}\left(f^2a_1 k_1 + k_ia_j\frac{\p^2 p_0}{\p x_i \p x_j}\right) = 0\,,\\
        \wt{\omega}^2 a_2 - a_j\frac{\p^2 p_0}{\p x_2 \p x_j} + \frac{k_2}{|\bk|^2}\left(f^2a_1 k_1 + k_ia_j\frac{\p^2 p_0}{\p x_i \p x_j}\right) = 0\,,
    \end{split}
\end{align}
which can be assembled into a matrix form
\begin{align}\label{eqn:a_matrix}
    \begin{pmatrix}
        \wt{\omega}^2 - f^2 + \dfrac{k_1 k_i}{|\bk|^2}\dfrac{\p^2 p_0}{\p x_i \p x_1} + \dfrac{f^2k_1^2}{|\bk|^2} - \dfrac{\p^2 p_0}{\p x_1 \p x_1} & \dfrac{k_1k_i}{|\bk|^2}\dfrac{\p^2 p_0}{\p x_i \p x_2} - \dfrac{\p^2 p_0}{\p x_1 \p x_2} \\
        \dfrac{k_2 k_i}{|\bk|^2}\dfrac{\p^2 p_0}{\p x_i \p x_1} + \dfrac{f^2k_1k_2}{|\bk|^2} - \dfrac{\p^2 p_0}{\p x_2 \p x_1} & \wt{\omega}^2 + \dfrac{k_2k_i}{|\bk|^2}\dfrac{\p^2 p_0}{\p x_i \p x_2} - \dfrac{\p^2 p_0}{\p x_2 \p x_2}
    \end{pmatrix}
    \begin{pmatrix}
        a_1\\a_2
    \end{pmatrix}
    = 0\,.
\end{align}
\begin{proposition}\label{prop:slice_WCI_disp_rel}
    The solvability condition for $a_1$ and $a_2$ implies the following dispersion relation
\begin{equation}\tag{\eqref{eqn:slice_WCI_dispersion_rel} revisited}
    \wt{\omega}^2 = \frac{f^2k_2^2}{|\bk|^2} + \left( \delta_{ij} - \frac{k_ik_j}{|\bk|^2} \right)\frac{\p^2p_0}{\p x_i\p x_j} \,.
\end{equation}
\end{proposition}
\begin{proof}
The solvability condition for $a_1$ and $a_2$ is that the determinant of this matrix found in equation \eqref{eqn:a_matrix} vanishes. To compute this determinant, we will separate it into terms multiplying $\wt{\omega}$, those multiplying $f$ (but not $\wt{\omega}$), and those featuring neither $f$ nor $\wt{\omega}$.
\begin{equation}
\begin{aligned}
    0 &= {\rm det}\begin{pmatrix}
        \wt{\omega}^2 - f^2 + \dfrac{k_1 k_i}{|\bk|^2}\dfrac{\p^2 p_0}{\p x_i \p x_1} + \dfrac{f^2k_1^2}{|\bk|^2} - \dfrac{\p^2 p_0}{\p x_1 \p x_1} & \dfrac{k_1k_i}{|\bk|^2}\dfrac{\p^2 p_0}{\p x_i \p x_2} - \dfrac{\p^2 p_0}{\p x_1 \p x_2} \\
        \dfrac{k_2 k_i}{|\bk|^2}\dfrac{\p^2 p_0}{\p x_i \p x_1} + \dfrac{f^2k_1k_2}{|\bk|^2} - \dfrac{\p^2 p_0}{\p x_2 \p x_1} & \wt{\omega}^2 + \dfrac{k_2k_i}{|\bk|^2}\dfrac{\p^2 p_0}{\p x_i \p x_2} - \dfrac{\p^2 p_0}{\p x_2 \p x_2}
    \end{pmatrix} \\
    &=: A[\wt{\omega},\wt{\omega}f] + B[f] + C[1] \,.
\end{aligned}
\end{equation}
We see that the terms corresponding to $A[\wt{\omega},\wt{\omega}f]$ produce the desired dispersion relation. We first demonstrate that the other terms coming from the determinant, $B[f]$ and $C[1]$, do not contribute. Beginning by evaluating terms which multiply the Coriolis parameter but do not feature $\wt{\omega}^2$, we have
\begin{align*}
    B[f] &= -\frac{f^2k_2^2}{|\bk|^2}\left( \frac{k_2k_i}{|\bk|^2}\frac{\p^2 p_0}{\p x_i \p x_2} - \frac{\p^2 p_0}{\p x_2 \p x_2} \right) - \frac{f^2k_1k_2}{|\bk|^2}\left( \frac{k_1k_i}{|\bk|^2}\frac{\p^2 p_0}{\p x_i \p x_2} - \frac{\p^2 p_0}{\p x_1 \p x_2} \right) 
    \\
    &= \frac{\p^2 p_0}{\p x_1 \p x_2}\left(\frac{f^2k_1k_2}{|\bk|^2} - \frac{f^2k_1k_2^3}{|\bk|^4} - \frac{f^2k_1^3k_2}{|\bk|^4}  \right) + \frac{\p^2p_0}{\p x_2\p x_2}\left(\frac{f^2k_2^2}{|\bk|^2} - \frac{f^2k_2^4}{|\bk|^4} - \frac{f^2k_1^2k_2^2}{|\bk|^4} \right)
    \\
    &= 0 \,.
\end{align*}
Similarly, we may evaluate the terms which feature neither $f$ nor $\wt{\omega}$
\begin{align*}
    C[1] &= \left( \frac{k_1 k_i}{|\bk|^2}\frac{\p^2 p_0}{\p x_i \p x_1} - \frac{\p^2 p_0}{\p x_1 \p x_1} \right)\left( \frac{k_2k_i}{|\bk|^2}\frac{\p^2 p_0}{\p x_i \p x_2} - \frac{\p^2 p_0}{\p x_2 \p x_2} \right) 
    \\
    &\qquad\qquad - \left( \frac{k_2 k_i}{|\bk|^2}\frac{\p^2 p_0}{\p x_i \p x_1} - \frac{\p^2 p_0}{\p x_2 \p x_1} \right)\left(\frac{k_1k_i}{|\bk|^2}\frac{\p^2 p_0}{\p x_i \p x_2} - \frac{\p^2 p_0}{\p x_1 \p x_2} \right) 
    \\
    &= \frac{\p^2p_0}{\p x_1 \p x_1}\frac{\p^2p_0}{\p x_2 \p x_2}\bigg( 1 - \frac{k_1^2}{|\bk|^2} - \frac{k_2^2}{|\bk|^2}  \bigg)
    \\
    &\quad + \frac{\p^2p_0}{\p x_1 \p x_1}\frac{\p^2p_0}{\p x_1 \p x_2}\bigg( \frac{k_1^3k_2}{|\bk|^4} - \frac{k_1k_2}{|\bk|^2} - \frac{k_1^3k_2}{|\bk|^4} + \frac{k_2k_1}{|\bk|^2} \bigg)
    \\
    &\quad + \frac{\p^2p_0}{\p x_1 \p x_2}\frac{\p^2p_0}{\p x_1 \p x_2}\bigg( \frac{k_1^2k_2^2}{|\bk|^4} -\frac{k_1^2k_2^2}{|\bk|^4} + \frac{k_2^2}{|\bk|^2} + \frac{k_1^2}{|\bk|^2} - 1 \bigg)
    \\
    &= 0 \,.
\end{align*}
Thus, we have $A[\wt{\omega},\wt{\omega}f] = 0$, and hence
\begin{align*}
    \wt{\omega}^4 &=  -\wt{\omega}^2\left( \frac{k_2k_i}{|\bk|^2}\frac{\p^2 p_0}{\p x_i \p x_2} - \frac{\p^2 p_0}{\p x_2 \p x_2} \right) - \wt{\omega}^2\left( - f^2 + \dfrac{k_1 k_i}{|\bk|^2}\dfrac{\p^2 p_0}{\p x_i \p x_1} + \dfrac{f^2k_1^2}{|\bk|^2} - \dfrac{\p^2 p_0}{\p x_1 \p x_1} \right)
    \\
    &= \wt{\omega}^2\frac{f^2k_2^2}{|\bk|^2} + \wt{\omega}^2\left( \delta_{ij} - \frac{k_ik_j}{|\bk|^2} \right)\frac{\p^2p_0}{\p x_i\p x_j} \,.
\end{align*}
This produces the dispersion relation \eqref{eqn:slice_WCI_dispersion_rel} after dividing by $\wt{\omega}^2$.
\end{proof}

\end{appendices}

\end{document}